\documentclass{sig-alternate-05-2015}
\usepackage{graphicx}
\usepackage{epsfig}
\usepackage{epstopdf}
\usepackage{amsmath}
\usepackage{stmaryrd}
\usepackage{amsfonts}
\usepackage{tikz}
\usepackage{graphics}
\usepackage{graphicx}

\usepackage{algorithm}
\usepackage{algorithmic}
\usepackage{graphicx}
\usepackage{subfigure}
\usepackage{url}
\usepackage{multirow}
\usepackage{eqparbox}
\usepackage{ifthen}
\usepackage{balance}
\usepackage{float}
\usepackage{makecell}
\usepackage{hyperref}
\usepackage{amssymb}

\usepackage{epsfig}
\usepackage{epstopdf}

\newtheorem{definition}{Definition}
\newtheorem{theorem}{Theorem}

\makeatletter
\renewcommand*{\@opargbegintheorem}[3]{\trivlist
      \item[\hskip \labelsep{\bfseries #1\ #2}] \textbf{(#3)}\ \itshape}
\makeatother

\newtheorem{lemma}{Lemma}
\newtheorem{concentration conditions}{Concentration Conditions}

\newtheorem{example}{Example}

\newcommand{\FMIN}{BTG-SM-CA }
\newcommand{\FMAX}{BTG-IM-CA }
\newcommand{\RHOMIN}{ADG-SM-CA }
\newcommand{\RHOMAX}{ADG-IM-CA }
\newcommand{\selectbyf}{SSBT}
\newcommand{\selectbyrhof}{SSAD}

\newcommand{\OnlyInFull}[1]{#1}
\newcommand{\OnlyInShort}[1]{}

\makeatletter
\renewcommand*{\@opargbegintheorem}[3]{\trivlist
      \item[\hskip \labelsep{\bfseries #1\ #2}] \textbf{(#3)}\ \itshape}
\makeatother

\newcommand{\I}{{\mathbb{I}}}

\begin{document}

\title{Cumulative Activation in Social Networks}
\numberofauthors{5}

\author{
\alignauthor
Xiaohan Shan\\
       \affaddr{Institute of Computing Technology, CAS}\\
       \email{shanxiaohan@ict.ac.cn}
\alignauthor
Wei Chen\\
       \affaddr{Microsoft}\\
       \email{weic@microsoft.com}
\alignauthor Qiang Li\\
       \affaddr{Institute of Computing Technology, CAS}\\
       \email{liqiang01@ict.ac.cn}
\and  
\alignauthor Xiaoming Sun\\
       \affaddr{Institute of Computing Technology, CAS}\\
       \email{sunxiaoming@ict.ac.cn}
\alignauthor Jialin Zhang\\
       \affaddr{Institute of Computing Technology, CAS}\\
       \email{zhangjl2002@gmail.com}
}
\maketitle
\begin{abstract}

Most studies on influence maximization focus on one-shot propagation, i.e. the influence is propagated from
	seed users only once following a probabilistic diffusion model and users' activation are determined via
	single cascade.
In reality it is often the case that a user needs to be cumulatively impacted
	by receiving enough pieces of information propagated to her 
	before she makes the final purchase decision.
In this paper we model such cumulative activation as the following process: first multiple pieces of information are propagated independently in the social network
	following the classical independent cascade model, then the user will be activated (and adopt the product) if the cumulative pieces of information she received
	reaches her cumulative activation threshold.
Two optimization problems are investigated under this framework: \emph{seed minimization with cumulative activation (SM-CA)}, which asks how to select a seed set with minimum size such that the number of cumulatively active
nodes reaches a given requirement $\eta$;
\emph{influence maximization with cumulative activation (IM-CA)}, which asks how to choose a seed set with fixed budget to maximize the number of cumulatively active nodes. 
For SM-CA problem, we design a greedy algorithm that yields a bicriteria $O(\ln n)$-approximation when $\eta=n$, where
	$n$ is the number of nodes in the network.
For both SM-CA problem with $\eta<n$ and IM-CA problem, we prove strong inapproximability results.
Despite the hardness results, we propose two efficient heuristic algorithms for SM-CA and IM-CA respectively based on the {\it reverse reachable set} approach. 
Experimental results on different real-world social networks show that our algorithms
significantly outperform baseline algorithms.
\end{abstract}


\section{Introduction}
With the wide popularity of social media and social network sites such as Facebook, Twitter, WeChat, etc.,
 social networks have become a powerful platform for spreading information, ideas and products among individuals.
In particular, product marketing through social networks can attract large number of customers.

Motivated by this background, influence diffusion in social networks has been extensively studied
	(cf. \cite{domingos2001mining,kempe2003maximizing,chen2013information}).
However, most of previous works only consider the influence after one-shot propagation --- influence propagates from
	the seed users only once and user activation or adoption is fully determined after this single cascade.
 In contrast, in the real world, people often make decisions after they have cumulated many pieces of information about
	 a new technology, new product, etc., and these different pieces of information are propagated in the network independently as different
	 information cascades.
	
 Consider the following scenario:
 A company is going to launch a new version (named as V11 for convenience) of their product with many new features,
 but most people are not familiar with these new features.
 Thus, it is often beneficial that the company conduct a series of advertisement and marketing campaigns covering different
	 features of the product.
An effective way of marketing in a social network is to select influential users as seeds to initiate the information cascades of these
	campaigns.
From potential customers' perspective,
	when they receive the first piece of information about V11 from their friends,
	they may find it interesting and forward it to their friends, but this may not necessarily lead to their purchase actions.
Later they may receive and be impacted by further information about V11, and when they are impacted by enough pieces of information cascades,
	they may finally decide to buy the new product.
	
We model the above behavior by an integrated process consisting of two phases:
	(a) repeated information cascades, and
	(b) threshold-based user adoptions.
First, there are multiple information cascades about multiple pieces of production information in the network.
We model information cascades by the classical independent cascade (IC) model \cite{kempe2003maximizing}:
A social network is modeled as a weighted directed graph, with an influence probability as
the weight on every edge.
Initially, some nodes are selected as seeds and become active, and all other nodes are inactive.
At each step, newly activated nodes have one chance to influence each of their inactive out-neighbors
	with the success probability given on the edge.
Independent cascade model is suitable to model simple contagions~\cite{centolaM07,chen2013information}
	such as virus and information propagation,
	and thus we adopt it to model information cascades in the first phase.
We consider multiple pieces of production information propagates independently following the IC model.
For the second phase, we assume that there is a threshold for each user, who
	will adopt the product if the amount of information
	that she receives in the first phase exceeds her threshold.
We measure the amount of information a user received as the fraction of information cascades that reaches the user,
	which is equivalent to the probability of the user being activated in an information cascade.
A node is {\em cumulatively activated} if this probability exceeds the threshold.
We refer to this model as the {\em cumulative activation (CA)} model.

Given the above cumulative activation model, the company may face one of the
	following two objectives:
	either the company has a fixed budget to activate the seed nodes, and wants
	to maximize the number of cumulative active nodes, or the company needs to reach
	a predetermined number of cumulative active nodes, and wants to minimize
	the number of seeds.

We formulate the above scenarios as the following two optimization problems: Seed minimization with cumulative activation (SM-CA) and influence maximization with cumulative activation (IM-CA). Given a directed graph with a probability on each edge and a threshold for each node, an activation requirement $\eta$ and a budget $k$, the SM-CA problem is to find a seed set
with minimum size such that the number of cumulatively activated nodes is at least $\eta$. The
IM-CA problem is to find a seed set with $k$ nodes such that the number of cumulatively activated nodes is maximum.

Let $\rho(S)$ denote the number of cumulative activated nodes given seed set $S$.
We first show that set function $\rho(\cdot)$ is not submodular,
 which means unlike most of the current studies, we cannot guarantee the approximation ratio by using the greedy algorithm directly.

For SM-CA problem, we consider the case $\eta=n$ and $\eta<n$ separately, where $n$ is
the number of nodes in the network and $\eta$ is the activation requirement.
The complexity results of these two cases are quite different.
 When $\eta=n$, we show while it is NP-hard to approximate SM-CA problem within factor $(1-\varepsilon)\ln n$  for any $\varepsilon>0$,
	 we can achieve a bicriteria $O(\ln n)$-approximation.
Our technique is to replace the nonsubmodular $\rho(S)$ with a submodular surrogate function $f(S)$, and show that the set
	of feasible solutions to the original SM-CA problem with constraint $\eta = n$ is exactly the same as the set of feasible
	solutions for $f(S)$ to assume its maximum value, and then we can apply the greedy algorithm to the surrogate $f(S)$
	instead of $\rho(S)$ to provide the theoretical guarantee.
When $\eta<n$, we construct a reduction from the {\em densest $k$-subgraph problem} to SM-CA problem and show that
SM-CA problem cannot be approximated within factor $\frac{1}{\sqrt{6}}n^{\delta/2}$
if the densest $k$-subgraph problem cannot be approximated within $n^\delta$, for any $\delta>0$,
	which is commonly acknowledged as
	a hard problem for some small $\delta$.

For IM-CA problem, we construct a reduction from the Set Cover problem and prove that it is NP-hard to approximate IM-CA problem within a factor of $n^{1-\varepsilon}$ for any $\varepsilon>0$.

Despite the approximation hardness on the SM-CA problem with $\eta<n$ and the IM-CA problem,
	we may still need practical solutions for them.
For this purpose, we propose some heuristic algorithms, which utilize the state-of-the-art approach in influence maximization,
	namely the \emph{Reverse Reachable Set (RR set)} approach \cite{borgs2014rrset,mtai2016sigmod,tang2014newrrset,tang2015rrset},
	to improve the efficiency of the algorithms comparing to the old greedy algorithms based on naive Monte Carlo simulations.

Finally, we conduct experiments on three real-world social networks to test the performance of our algorithms.
Our results demonstrate that one heuristic algorithm proposed consistently out-performs all other algorithms under comparison in
	all test cases and clearly stands out as the winning choice for both the SM-CA and IM-CA problems.

To summarize, our contributions include:
 (a) we propose the seed minimization and influence maximization problem under cumulative
activation (SM-CA problem and IM-CA problem respectively),
	which is a reasonable model for purchasing behavior of customers exposed to repeated information cascades;
 (b)we design an $O(\ln n)$ approximate algorithm for SM-CA problem when $\eta=n$;
 (c) we show strong hardness results for SM-CA problem with $\eta<n$ and IM-CA problem;
 (d) we propose efficient heuristic algorithms and validate them through extensive experiments on real-world datasets and
	 conclude that one heuristic is the best choice for both SM-CA and IM-CA problems.
%
%

\subsection{Related Work}
The classical influence maximization problem is to find a seed set of at most $k$ nodes to maximize the expected number of active nodes.
 It is first studied as an algorithmic problem by Domingos and Richardson \cite{domingos2001mining}
and Richardson and Domingos~\cite{richardson2002mining}. Kempe et al. \cite{kempe2003maximizing}
first formulate the problem as a discrete optimization problem.
 They summarize the independent cascade model and the linear threshold model,
and obtain approximation algorithms for influence maximization by applying submodular
	function maximization.
 Extensive studies follow their approach and provide more efficient algorithms
 \cite{ChenWY09efficientinfluence,chen2010sharpphard,Leskovec2007costeffective}.
 Leskovec et al. \cite{Leskovec2007costeffective} present a ``lazy-forward" optimization method in selecting new seeds,
 which greatly reduce the number of influence spread evaluations. Chen et al.
	 \cite{ChenWY09efficientinfluence,chen2010sharpphard} propose scalable algorithms
which are faster than the greedy algorithms proposed in \cite{kempe2005influential}. Recently, Borgs et al. \cite{borgs2014rrset} and Tang et al. \cite{tang2015rrset,tang2014newrrset} and Nguyen et al. \cite{mtai2016sigmod} propose a series of more effective algorithms for influence maximization in large social networks that both has theoretical guarantee and practical efficiency.
The approach is based on the ``Reverse Reachable Set'' idea first proposed in \cite{borgs2014rrset}.

Another aspect of influence problem is seed set minimization, Chen \cite{chen2009approximability} studies the seed minimization
problem under the fix threshold model and shows some strong negative results for this model.
 Long et al. \cite{long2011minimizing} also study independent cascade model and linear threshold model
from a minimization perspective.
 In \cite{goyal2012minimizing}, Goyal et al. study the problem of finding the
minimum size of seed set such that the expected number of active nodes reaches a given threshold,
 they provide a bicriteria approximation algorithm for this problem.
 Zhang et al. \cite{zhang2014probabilistic} study the seed set minimization problem
with probabilistic coverage guarantee, and design an approximation algorithm for this problem.
 He et al. \cite{he2014positive} study positive influence model under single-step
activation and propose an approximation algorithm. Note that, the work in \cite{he2014positive} is a special case of our work.

Beyond influence maximization and seed minimization,
 another interesting direction is the learning of social influence over
real online social network data set, e.g.
	influence learning in blogspace \cite{Gruhl2004blogspace} and academic collaboration
network \cite{tangjie2009largescale}.

Most early studies on influence maximization and influence learning are summarized in the monograph \cite{chen2013information}.
However, almost all the existing studies consider only node activation after a single information or influence cascade.
Our work differentiate with all these studies on this important aspect, as discussed in the introduction.

\emph{Paper organization.} We formally define the diffusion model and the optimization problems SM-CA and IM-CA
in Section \ref{sec:model}. The approximation algorithms and hardness results of these two problems
are proposed in Section \ref{sec:algorithmsandhardness}, including a greedy algorithm
for SM-CA problem with $\eta=n$ in section \ref{sec:eta=n}, the hardness result of SM-CA problem with $\eta<n$ in Section \ref{sec:eta<n}
and the inapproximate result of IM-CA problem in Section \ref{sec:im-ca}.
 In Section \ref{sec:heuristicalgorithms}, we present two heuristic algorithms for SM-CA problem and two heuristic algorithms for IM-CA problem.
 Section \ref{sec:experiments} shows our experimental results on real-world datasets.
 We summarize the paper with some further directions in Section 6.

\section{Model and Problem Definitions}\label{sec:model}

Our social network is defined on a directed graph $G=(V, E)$, where $V$ is the set of nodes representing individuals
and $E$ is the set of directed edges representing social ties between pairs of individuals.
 Each edge $e=(u, v)\in E$ is associated with an influence probability $p_{uv}$, which represents the probability that $u$
	 influences $v$.

The entire activation process consists of information diffusion process and node activation.
The information diffusion process follows the independent cascade (IC) model proposed by Kempe et al. \cite{kempe2003maximizing}.
 In the IC model, discrete time steps $t=0, 1, 2, \cdots$ are used to model the diffusion process. Each node in $G$ has two states:
 inactive or active, At step 0, a subset $S\subseteq V$ is selected as seed set and
	 nodes in $S$ are active directly, while nodes not in $S$ are inactive.
For any step $t\ge 1$, if a node $u$ is newly active at step $t-1$,
 then $u$ has a single chance to influence each of its inactive out-neighbor $v$ with independent probability $p_{uv}$
	 to make $v$ active.
 Once a node becomes active, it will never return to the inactive state.
 The diffusion process stops when there is no new active nodes at a time step.

The above basic IC model describe the diffusion of one piece of information, but actually there could be many pieces of information
	about a product being propagated in the network, all following the same IC model.
Users' final production adoption is based on cumulative information collected, which we refer to as {\em cumulative activation (CA)}
	and is described below, and it is different from the user becoming active for one piece of information specified above in the IC model.
Let $P_u(S)$ be the probability that $u$ becomes active after an information cascade starting from the seed set $S$.
Since $P_u(S)$ also represents the fraction of information accepted by $u$ in multiple cascades, we use $P_u(S)$ to define
	cumulative activation:
Suppose that each node $u\in V$ has an activation threshold $\tau_u \in (0,1]$, then $u$ becomes \emph{cumulatively active} if $P_u(S)\geq \tau_u$.
Given a target set $U\subseteq V$ and a seed set $S$, let $\rho_U(S)$ be the number of cumulatively active nodes in $U$ from seed set $S$.
When $U=V$, we omit the subscript $U$ and use $\rho(S)$ directly.

We consider two optimization problems under cumulative activation, seed minimization with cumulative activation (SM-CA)
and influence maximization with cumulative activation (IM-CA).
 SM-CA aims at finding a seed set $S$ with minimum size such that there are at least $\eta$ $(\eta\leq n)$ nodes in the target set become cumulatively active.
 IM-CA is the problem of finding a seed set of size $k$ to maximize
the number of cumulatively active nodes in the target set. The formal definitions are as follows.

\begin{definition}[Seed minimization with cumulative activation]
In the seed minimization with cumulative activation (SM-CA) problem,
 the input includes a directed graph $G=(V,E)$ with $|V|=n,~|E|=m$,
 an influence probability vector $P=\{p_{uv}: p_{uv}\in [0, 1],~(u, v)\in E\}$,
a target set $U\subseteq V$, an activation threshold $\tau_u\in (0, 1]$ for each node $u\in U$ and a coverage requirement $\eta\leq |U|$.
 Our goal is to find the minimum size seed set $S^*\subseteq V$ such that at least $\eta$ nodes in U can be cumulatively activated,
that is,
\begin{center}
$S^*=\underset{S:\rho_U(S)\geq \eta}{\arg \min}~|S|$.
\end{center}
\end{definition}

\begin{definition}[Influence maximization with cumulative activation]
In the influence maximization with cumulative activation (IM-CA) problem,
 the input includes a directed graph $G=(V, E)$ with $|V|=n,~|E|=m$,
 an influence probability vector $P=\{p_{uv}: p_{uv}\in [0, 1],~(u,v)\in E\}$,
 a target set $U\subseteq V$, an activation threshold $\tau_u\in (0, 1]$ for each node $u$ and a size budget $k\leq n$.
 Our goal is to find a seed set $S^*\subseteq V$ of size $k$ such that the number of cumulatively active nodes in $U$ is maximized,
 that is,
\begin{center}
$S^*=\underset{S:|S|=k}{\arg \max}~\rho_U(S)$.
\end{center}
\end{definition}

\subsection{Equivalence to Frequency-based Definition } \label{sec:furtherdiscussion}
Suppose there are $N$ diffusions, which lead to final cumulative activation.
Intuitively, a node $u\in V$ becomes cumulative activated when
    the number of times that $u$ becomes influenced in these $N$ diffusions is larger than a threshold.
Formally, given a seed set $S$ and a node $u\in V$, let $X_u^i$ be a random variable defined as follows:
    $X_u^i(S)=1$ if $u$ is influenced in the $i$-th diffusion and $X_u^i(S)=0$ otherwise.
Thus, \begin{small}$X_u(S)=\sum_{i=1}^{N}X_u^i(S)$\end{small} denotes the number of times that $u$ becomes influenced after $N$ diffusions, and
	 $X_u(S)/N$ is the influence frequency of $u$.
By Hoeffding's inequality, we show the relationship between $p_u(S)> \tau_u$ and $X_u(S) /N \geq \tau_u$ in Lemma \ref{lem:modeldiscussioncumulative}.
\begin{lemma}\label{lem:modeldiscussioncumulative}
Given a seed set $S$, a node $u\in V$ and  a large enough $N$,
	(a) if $p_u(S)> \tau_u$, then \begin{small}$\Pr(X_u(S)/N\geq \tau_u)=1-o(1)$\end{small}; and
	(b) if $p_u(S)< \tau_u$, then \begin{small}$\Pr(X_u(S)/N\leq \tau_u)=1-o(1)$\end{small}, where $o(1)$ is asymptotic
	to the number of diffusions $N$.
\end{lemma}
\begin{proof}
It is obvious that the expectation of $X_u(S)$ is $E[X_u(S)]=\sum_{i=1}^{N}E[X_u^i(S)]=Np_u(S)$.
When $P_u(S)< \tau_u$, by Hoeffding's inequality, we have:
\begin{equation*}
\begin{aligned}
&\Pr (X_u(S)/N\geq \tau_u)\\
=&\Pr(X_u(S)-E[X_u(S)]\geq N\tau_u-E[X_u(S)])\\
\leq& \exp(-2N(N\tau_u-E[X_u(S)])^2)\\
=&\exp(-2N^3(\tau_u-P_u(S))^2)=o(1).\\
\end{aligned}
\end{equation*}
Thus, when $N$ is large enough, $X_u(S)/N\leq \tau_u$ is a high probability event if $p_u(S)< \tau_u$.
Similarly, $X_u(S)/N\geq \tau_u$ is a high probability event if $p_u(S)> \tau_u$.
\end{proof}

Based on Lemma \ref{lem:modeldiscussioncumulative}, the formal definition of cumulative activation
    is consistent with our motivation.
\subsection{Comparison with IC and LT Models}

We first explain the differences between the CA model the IC model.
CA model uses IC model as information cascades in its first stage, and thus the main difference is at the
	determination of which nodes are finally activated, or simply at the objective function.
This is clearly illustrated by the simple example in Figure~\ref{comparision of CA and IC}, which shows a five-node graph
	with edge probabilities shown next to edges.
In the IC model, it is clear that the influence spread of $v_1$ and $v_2$ are the same:
	$\sigma(\{v_1\}) = \sigma(\{v_2\}) = 2.2$.
For the CA model, if every node has the activation threshold as $0.3$, then $\rho(\{v_1\}) = 2$ and
	$\rho(\{v_2\}) = 4$, because $v_1$ can only activate itself and $u_1$, while $v_2$ can activate
	itself and $u_1,u_2,u_3$.
If the activation threshold of every node is increased to $0.6$, then $\rho(\{v_1\}) = 2$ and
	$\rho(\{v_2\}) = 1$.
Therefore, if we want to select one seed in the influence maximization task, for IC model either $v_1$
	or $v_2$ is fine, but for the CA model $v_1$ or $v_2$ is selected based on different threshold values.
This means the influence maximization task under CA model is different from the task under the IC model.
The above example also provides the intuition that the influence maximization task under the IC model
	focuses on the average effect of the influence, while the task under the CA model may need to
	select either nodes that has wide but average influence (e.g. $v_2$) or nodes with concentrated
	influence (e.g. $v_1$) based on the threshold setting.

\begin{figure*}[t]
  \subfigure[CA \& IC]
  {
    \label{comparision of CA and IC}
    \includegraphics[width=4cm]{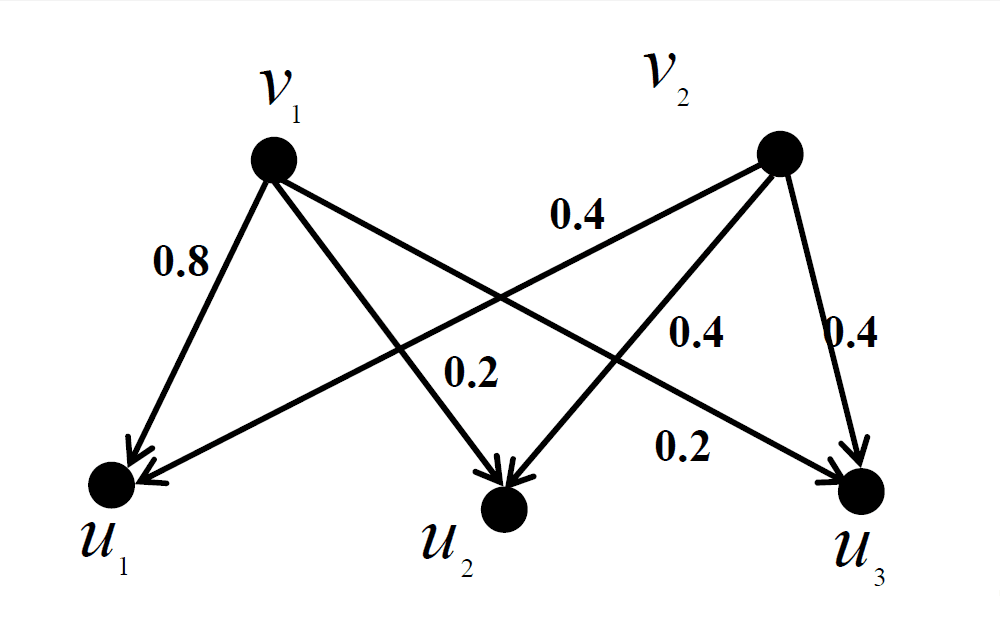}\\
  }
  \subfigure [nonsubmoduarity]
  {
    \includegraphics [width=3.5cm]{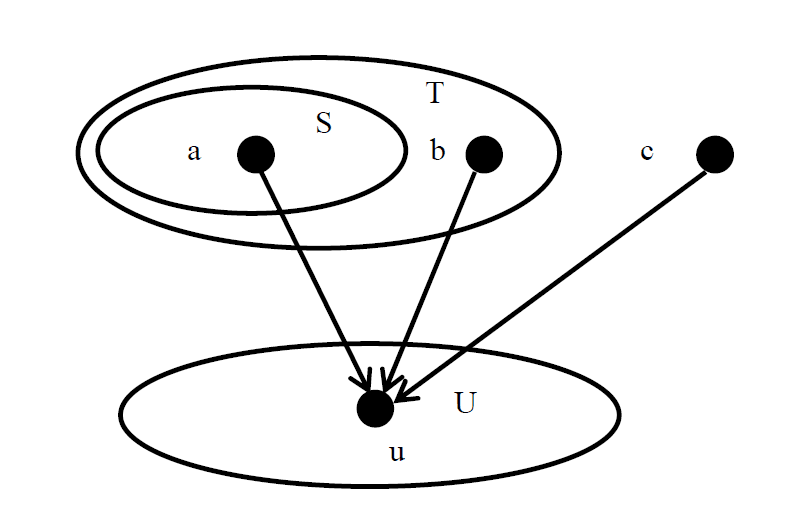} \\
    \label{fig:nonsub}
  }
  \subfigure [$\eta<n$]
  {
    \includegraphics [width=3.8cm]{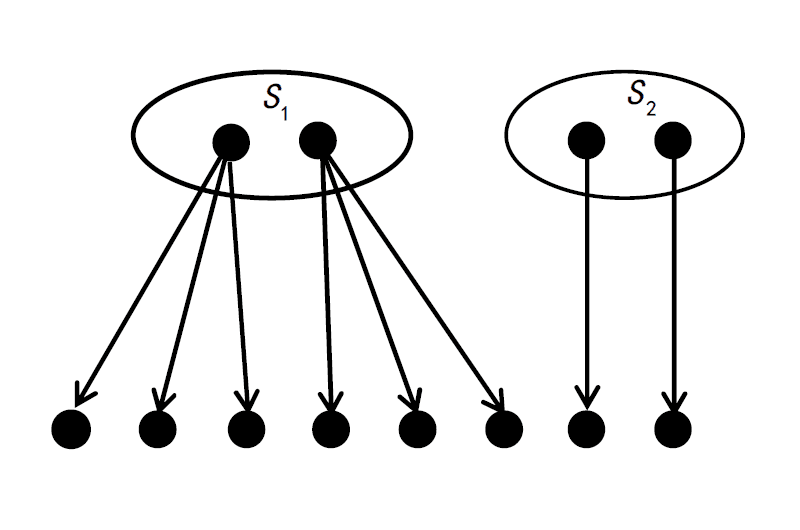} \\
    \label{fig:etaln}
  }
  \caption {Figures for understanding the model}
  \label{fig:model undersanding}
\end{figure*}


We next distinguish our CA model with the
	popular linear threshold (LT) model proposed in~\cite{kempe2003maximizing}.
In the LT model, each edge $(u,v)$ has a weight $w_{uv}\in [0,1]$ with $\sum_u w_{uv} \le 1$
($w_{uv}=0$ if $(u,v)$ is not an edge).
Each node $v$ has a threshold $\theta_v$, which is drawn from $[0,1]$ uniformly at random before the
propagation starts.
Then, starting from the seed set $S$, an inactive node $v$ becomes active at time $t\ge 1$ if any only if
the total weights from its active in-neighbors exceeds $v$'s threshold:
$\sum_{u\in S_{t-1}} w_{uv} \ge \theta_v$, where $S_t$ is the set of active nodes at time $t$
with $S_0 = S$.

Despite the superficial similarity on using thresholding to model user adoption behavior, the two models
	are quite different.
One key difference is that in the LT model what is being propagated are the user adoption behavior, while
	in the CA model, what is being propagated are multiple pieces of information about a product, and user's
	adoption in the end is based on the information received.
This is actually the difference between CA and most other models on influence diffusion, as discussed in
	the introduction.
This further leads to a specific difference between LT and CA:
	the threshold in the LT model is on the number of friends who already
	adopt a product, while the threshold in the CA model is on the fraction of information cascades
	that reach a user.
Finally, in LT the threshold $\theta_v$ is a random number in $[0,1]$, making the influence spread
	objective function submodular, while in CA the threshold $\tau_v$ is fixed as an input, causing
	the objective function $\rho(S)$ not submodular, as discussed in the next section.

%
\section{\leftline {Algorithms and hardness results}}
\label{sec:algorithmsandhardness}

In this section, we provide algorithmic as well as hardness results for SM-CA problem and IM-CA problem.

A set function $f:2^V\rightarrow R$ is {\em monotone} if
	$f(S) \le f(T)$ for all $S\subseteq T \subseteq V$, and
	{\em submodular} if $f(S\cup\{w\})-f(S)\geq f(T\cup\{w\})-f(T)$ for all $S\subseteq T\subseteq V$ and $w\notin T$.
It is well known that monotone submodular functions leads to a good approximation ratio by using the greedy algorithm
	\cite{NWF78}, and indeed
	 most of the existing work on social influence takes advantage of this nature
(e.g. \cite{borgs2014rrset,chen2015ec,goyal2012minimizing,kempe2003maximizing}).

Unfortunately, our objective function $\rho_U(S)$ is monotone but not submodular in general as shown in the example
	below, which makes our problems much harder.

\begin{example} (Figure \ref{fig:nonsub})\label{eg:rho(S)not submodular}
 Suppose $G$ is a bipartite graph and the influence probability on each edge is $1/2$,
 the activation thresholds are: $\tau_u=7/8$, $\tau_a=\tau_b=\tau_c=1$.
 Let $U=\{u\}$.
 Let $S=\{a\}$ and $T=\{a, b\}$,
 then, $\rho_U(S)=0$, $\rho_U(T)=0$, $\rho_U(S\cup \{c\})=0$, $\rho_U(T\cup \{c\})=1$.
 Thus, $\rho_U(S\cup \{c\})-\rho_U(S)=0<\rho_U(T\cup \{c\})-\rho_U(T)=1$, implying that
 $\rho_U(S)$ is not submodular. We further remark that in this example, if we set $U=V$,
	 the function $\rho(S)=\rho_V(S)$ is still not submodular.
 \end{example}


In the rest part of this section, we consider how to design approximation algorithms for SM-CA problem and IM-CA problem as well as the computational complexity of them.

\subsection{Seed minimization with cumulative activation (SM-CA) problem}
\label{sec:sm-ca}
In this section, we study SM-CA problem. We first show the hardness result of SM-CA problem in Theorem \ref{thm:hardness for SM-CA}.

\begin{theorem}\label{thm:hardness for SM-CA}
SM-CA problem is NP-hard. Moreover,
 SM-CA problem cannot be approximated within $(1-\varepsilon)\ln \eta$ in polynomial time unless $NP\subseteq DTIME(n^{O(\log{\log n})})$, $\forall~\varepsilon>0$.
\end{theorem}
\OnlyInShort{
\begin{proof}[sketch]
The proof is by an objective-value-preserving reduction from the Partial Set Cover (PSC) problem, which is shown by Feige that
	it cannot be approximated within a factor of $(1-\varepsilon)\ln \eta$ $(\forall~ \varepsilon>0)$ in polynomial time unless $NP\subseteq DTIME(n^{O(\log{\log n})})$ \cite{feige98setcoverhardness}.
\end{proof}
}

\OnlyInFull{
\begin{proof}
We construct a reduction from the Partial Set Cover (PSC) problem. An instance of PSC problem $I_{PSC}=(U,\cal{S},\eta)$ consists of a ground set $U$ and a family of subsets $\cal{S}\subseteq$ $2^U$ and a coverage requirement $\eta\leq |U|$. The objective is to find a subcollection $\cal{C}\subseteq \cal{S}$ such that $|\bigcup\cal{C}|\geq \eta$ and $|\cal{C}|$ is minimized.

Given any instance of PSC problem $I_{PSC}=(U,\cal{S},\eta)$, we construct an instance of SM-CA problem $I_{SM-CA}=(G,\eta)$ as follows.
The PSC instance is reduced to a bipartite graph $G=(V,W,E)$, in which each node $v\in V$ corresponds to a subset $C_v\in\cal{S}$ one to one, each node $w\in W$ corresponds to an element $u_w\in U$ one to one, there is an edge $(v,w)\in E$ if and only if $u_w\in C_v$. The target set is $W$, the influence probability on each edge is 1, the activation thresholdof each node in $W$ is 1.
The activation requirement is $\eta$, which is the same with the coverage requirement in the PSC instance.

Based on the above construction, it is easy to check that the objective function values for any given S in the instances
	$I_{SM-CA}$ and $I_{PSC}$ are always the same.
Which means the two problems should have the same approximation ratio.
For the PSC problem, Feige showed that it cannot be approximated within a factor of $(1-\varepsilon)\ln \eta$ $(\forall \varepsilon>0)$ in polynomial time unless $NP\subseteq DTIME(n^{O(\log{\log n})})$ \cite{feige98setcoverhardness}. Therefore, SM-CA problem has the same computational complexity.
\end{proof}
}
Based on the hardness result of SM-CA problem, our next goal is to design an algorithm with approximation ratio close to $\ln \eta$.
Surprisely, it turns out that the results are quite different between ``activating all nodes'' ($\eta=n$) and ``partial activation'' ($\eta<n$), as we discuss separately below.
\subsubsection{SM-CA problem with $\eta=n$}\label{sec:eta=n}



When $\eta=n$, we can design an algorithm with a bicriteria $O(\ln n)$-approximation, even though the objective function is not submodular.
The key idea to our solution is to find a submodular function $f(S)$ as the surrogate for the original nonsubmodular $\rho(S)$,
	as the following lemma specifies.

\begin{lemma}\label{lem:feasible:eta=n}
When $\eta=n$, a seed set $S$ is a feasible solution to the SM-CA problem if and only if $f(S)=\sum_{u\in V}\tau_u$, where $f(S)$ is
	a {\em surrogate function} defined as:
\begin{center}
$f(S)=\sum_{u\in V}\min\{P_u(S),\tau_u\}$.
\end{center}
\end{lemma}
\begin{proof}
If $S$ is a feasible solution to SM-CA, that is, $\rho(S)=n=\eta$, then every node $u$ satisfies $P_u(S)\ge \tau_u$, and thus
	$f(S)= \sum_{u\in V}\tau_u$.
The only-if part is also straightforward.
\end{proof}


The above lemma implies that minimizing seed set size for the constraint of $\rho(S) = n$ is the same as minimizing the seed set size
	for the constraint of $f(S)=\sum_{u\in V}\tau_u$.
The reason we want to switch the minimization problem on the surrogate function $f(S)$ is because it is submodular, as pointed out by
	Lemma~\ref{thm:f(s)submodular}.
We remark that in \cite{NIPS2014}, Farajtabar et al. study an objective function in the similar form as $f(S)$ in the continuous-time
	influence model, but the interpretation
	of $\tau$ is the cap on user activity intensity in \cite{NIPS2014} rather than the activation threshold.

\begin{lemma}\label{thm:f(s)submodular}
The surrogate function $f(S)$ is monotone and submodular.
\end{lemma}
\begin{proof}[sketch]
It is obvious that $f(S)$ is monotone.
For submodularity, following \cite{kempe2003maximizing} we know that $P_u(S)$ is submodular.
Then it is easy to check that the minium of a submodular function and a constant is still submodular, and the simple summation of
	submodular functions is also submodular.
\end{proof}
Having the submodularity, we can design a greedy algorithm guided by $f(S)$.
 But like most work in the IC model, we cannot avoid the problem of computing $f(S)$.
 It has been shown that exactly computing $\sigma(S)$ in the IC model is \#P-hard \cite{chen2010sharpphard},
 where $\sigma(S)=\sum_{u\in V}P_u(S)$ is the expected number of active nodes given the seed set $S$.
 Thus, computing $f(S)$ is also \#P-hard since $\sigma(S)=\sum_{u\in V}P_u(S)=f(S)$ if we set $\tau_u=1$ for all $u\in V$.
In this section, we use Monte Carlo simulation to estimate $f(S)$. A more efficient method will be discussed in Section \ref{sec:improved gerrdy algorithms based on RRS}.
\begin{algorithm}[t]
\renewcommand{\algorithmicrequire}{\textbf{Input:}}
\renewcommand\algorithmicensure {\textbf{Output:}}
    \caption{Estimate $f(S)$ by Monte Carlo} \label{alg:compf(s)}
    \begin{algorithmic}[1]
    \REQUIRE ~~ $G=(V,~E),~\{p_{u v}\}_{(u,~v)\in E}, \{\tau_u\}_{u\in V},~U,~S,~R$\\
    \ENSURE ~~  $\hat{f}(S)$: the estimation of $f(S)$\\
\STATE  $\hat{f}(S)=0$;
\STATE $\hat{P}_u(S)=0$; $t_u=0$ for all $u \in V$
\FOR {$i=1$ to $R$}
\STATE simulate IC diffusion from seed set $S$
\IF {$u$ is activated}
\STATE  $t_u=t_u+1$
\ENDIF
\ENDFOR
\FOR {$u\in U$}
    \STATE $\hat{P}_u(S)=t_u/R$
    \IF {$\hat{P}_u(S)\geq\tau_u$}
        \STATE  $\hat{f}(S)=\hat{f}(S)+\tau_u$
    \ELSE
        \STATE  $\hat{f}(S)=\hat{f}(S)+\hat{P}_u(S)$
    \ENDIF
\ENDFOR
\RETURN $\hat{f}(S)$
\end{algorithmic}
\end{algorithm}

\begin{algorithm}[t!]
\renewcommand{\algorithmicrequire}{\textbf{Input:}}
\renewcommand\algorithmicensure {\textbf{Output:}}
    \caption{Greedy algorithm for SM-CA with $\eta=n$} \label{alg:greedy-min-eta=n}
    \begin{algorithmic}[1]
    \REQUIRE ~~ $G=(V,~E),~\{p_{uv}\}_{(u,~v)\in E},\{\tau_u\}_{u\in V},~U$, $\varepsilon$\\
    \ENSURE ~~  Seed set $S$\\
\STATE  $S=\emptyset$, $\hat{f}(S)=0$
   \WHILE {$\hat{f}(S)<\sum_{u\in V}\tau_u-\varepsilon$}
        \STATE choose $v=\arg\max_{u\in V\setminus S}~[\hat{f}(S\cup \{u\})-\hat{f}(S)]$
        \STATE $S=S\cup \{v\}$
    \ENDWHILE
\RETURN $S$
\end{algorithmic}
\end{algorithm}
Algorithm \ref{alg:compf(s)} shows the procedure of the Monte Carlo method. Given a seed set $S$ and a node $u$, Algorithm \ref{alg:compf(s)} simulates the diffusion process from $S$ for $R$ runs,
 and uses the frequency that $u$ has been influenced as the estimation of $P_u(S)$.
 Then we can obtain the estimation of $f(S)$ directly by a truncation operation. The estimations of $P_u(S)$ and $f(S)$ are denoted by $\hat{P}_u(S)$ and $\hat{f}(S)$ respectively.

The accuracy of the estimate $\hat{f}(S)$ depends on the number of simulation runs $R$, as rigorously specified by
	the following lemma.
\begin{lemma}\label{lem:error estimate}
For any seed set $S$, suppose $\hat{f}(S)$ is the estimate of $f(S)$ output by Algorithm \ref{alg:compf(s)},
then $\forall~\gamma>0,~\delta>0$, $\Pr(|\hat{f}(S)-f(S)|\leq \gamma)\geq 1-1/n^\delta$ if $R\geq (n^2\ln (2n^{\delta+1}))/ 2\gamma^2$.
\end{lemma}

\begin{proof}
For each node $u$, Let $X_u=\sum_{i=1}^R X_{u}^{(i)}$, where $X_{u}^{(i)}$ is a random variable
defined as $X_{u}^{(i)}=1$ if $u$ is influenced in the $i$-th simulation and $X_{u}^{(i)}=0$ otherwise.
Then $X_u$ is the number of times that $u$ is active after $R$ simulations.
 Thus, $X_u=R\cdot\hat{P}_u(S)$ and $E[X_u]=\sum_{i=1}^R E[X_{u}^{(i)}]=R\cdot P_u(S)$.
By Hoeffding's inequality and the condition $R\geq (n^2\ln (2n^{\delta+1}))/ 2\gamma^2$, for any constant $\gamma>0$ and $\delta>0$,
\begin{align*}
&\Pr(|\hat{P}_u(S)-P_u(S)|\geq\gamma/n)=\Pr(|X_u-E[X_u]|\geq R\gamma/n) \\
&\leq 2\exp\left(-\frac{2(\frac{R\gamma}{n})^2}{R} \right)\leq\frac{1}{n^{\delta+1}}.
\end{align*}

We next show that $|\hat{f}(S)-f(S)|\leq \sum_{u\in V}|\hat{P}_u(S)-P_u(S)|$ always holds.
\begin{equation*}
\begin{aligned}
&|\hat{f}(S)-f(S)|
=\left|\sum_{u\in V}\min\{\hat{P}_u(S),~\tau_u\}-\sum_{u\in V}^{}\min\{P_u(S),~\tau_u\}\right|\\
&\le  \sum_{u\in V} \left|\min\{\hat{P}_u(S),~\tau_u\}- \min\{P_u(S),~\tau_u\}\right|\\
&\leq \sum_{u\in V}|\hat{P}_u(S)-P_u(S)|.
\end{aligned}
\end{equation*}
Then, we have
\begin{align*}
&\Pr(|\hat{f}(S)-f(S)|\leq\gamma) \ge \Pr(\sum_{u\in V}|\hat{P}_u(S)-P_u(S)| \le \gamma) \\
&\ge \Pr(\forall u\in V, |\hat{P}_u(S)-P_u(S)| \le \gamma/n) \\
& = 1 - \Pr(\exists u\in V, |\hat{P}_u(S)-P_u(S)| \ge \gamma/n) \\
& \ge 1 - \sum_{u\in V} \Pr(|\hat{P}_u(S)-P_u(S)| \ge \gamma/n) \ge 1-\frac{1}{n^\delta}.
\end{align*}
\end{proof}
Having the estimation algorithm of $f(S)$, we show our greedy algorithm for SM-CA problem with $\eta=n$ in Algorithm \ref{alg:greedy-min-eta=n}.

Algorithm \ref{alg:greedy-min-eta=n} starts from an empty seed set $S$. At each iteration, it adds one node $v$
providing the largest marginal increment to $\hat{f}(S)$ into $S$, i.e.,
$v=\underset{u\in V}{\arg\max}~[\hat{f}(S\cup \{u\})-\hat{f}(S)]$.
The algorithm ends when $\hat{f}(S)\geq \sum_{u\in V}\tau_u-\varepsilon$ and outputs $S$
as the selected seed set. Goyal et al. proved the performance guarantee for the greedy algorithm
when $f(S)$ is monotone and submodular \cite{goyal2012minimizing}.

\begin{theorem}[\cite{goyal2012minimizing}]\label{thm:goyal}
Let $G=(V,~E)$ be a social graph and $f(\cdot)$ be a nonnegative, monotone and submodular function defined on
$2^V$.
Given a threshold $0<\eta\leq f(V)$, let $S^*\subseteq V$ be a subset with minimum size such that
$f(S^*)\geq\eta$, and $S$ be the greedy solution using a $(1-\gamma)$-approximate function
$\hat{f}(\cdot)$ with the stopping criteria $\hat{f}(S)\geq\eta-\varepsilon$.
 Then, there exists a $\gamma$ such that for any $\varphi>0$ and $\varepsilon>0$,
 $|S|\leq|S^*|(1+\varphi)(1+\ln(\eta/\varepsilon))$ with high probability.
\end{theorem}
Now we can conclude the approximation ratio of Algorithm \ref{alg:greedy-min-eta=n} based on
Lemmas~\ref{lem:feasible:eta=n}--\ref{lem:error estimate} and Theorem \ref{thm:goyal}.

\begin{theorem}\label{thm:approratio}
When $\eta=n$, for any $\phi>0,~\varepsilon>0$,
 Algorithm \ref{alg:greedy-min-eta=n} ends when $\hat{f}(S)\ge \sum_{u\in V}\tau_u-\varepsilon$ and approximates SM-CA problem within
a factor of $(1+\phi)\cdot(1+\ln \frac{\sum_{u\in V}\tau_u}{\varepsilon})$ with high probability.
\end{theorem}

\subsubsection{SM-CA problem with $\eta<n$}\label{sec:eta<n}

When $\eta<n$, the surrogate function $f(S)$ does not enjoy the property in Lemma~\ref{lem:feasible:eta=n} any more, and thus
	the problem becomes more difficult.
We use the following example to explain this phenomenon.

\begin{example}(Figure \ref{fig:etaln}). Suppose the influence probability on each edge from $S_1$ is 0.5 and each edge from $S_2$ is 1. The activation threshold of each node is 1, and $\eta=3$.
	Then $f(S_1)=5$, $f(S_2)=4$, but $\rho(S_1)=2$, $\rho(S_2)=4$. Thus, $S_1$ is not a feasible solution even though $f(S_1)$ is large enough. This simple example shows that too many ``small active probability'' nodes may mislead $f(S)$ causing it
	to diverge significantly from $\rho(S)$.
\end{example}


Now we show the hardness result of SM-CA problem with $\eta<n$. Our analysis is based on the hardness of
the \emph{densest $k$-subgraph} (D$k$S) problem \cite{feige2001denseksubgraph}.
An instance of D$k$S problem consists of an undirected graph $G=(V, E)$ and a parameter $k<n$, where $n=|V|$.
The objective is to find a subset $V'\subseteq V$ of cardinality $k$ such that
 the number of edges with both endpoints in $V'$ is maximized.

The first polynomial time approximation algorithm for D$k$S problem is given by
Feige et al. in 2001 \cite{feige2001denseksubgraph} with the approximation ratio $O(n^{1/3})$.
This result was improved to $O(n^{1/4+\varepsilon})$ (for any $\varepsilon>0$) by Bhaskara et al. \cite{bhaskara2010dks1/4} in 2010 and this is the currently the best known guarantee. For the hardness of D$k$S problem, Khot \cite{Khot2006ptas} proved that the D$k$S problem does not admit PTAS
under the assumption that NP problems does not have sub-exponential time randomized algorithms.
 The exact complexity of approximating D$k$S problem is still open, but it is widely believed that
D$k$S problem can only be approximated within polynomial ratio.

Partially borrowing the idea in \cite{Hajiaghayi2006polyhardness},
we can prove a hardness result for SM-CA problem with $\eta<n$ based on the hardness of D$k$S problem.
\OnlyInShort{
The whole proof of the following theorem is reported in our full version \cite{fullversion}.
}
\begin{theorem}\label{thm:hardness for eta<n}
When $\eta<n$, SM-CA problem cannot be approximated within $\frac{1}{\sqrt{6}}n^{\delta/2}$
if D$k$S problem cannot be approximated within $n^\delta$, for any $\delta>0$.
\end{theorem}

\OnlyInFull{
\begin{proof}
Suppose there is a polynomial time approximation algorithm $\mathcal A$ with performance ratio $r$ for SM-CA with $\eta<n$, we design an algorithm for D$k$S problem based on $\mathcal A$, which has approximation ratio $6r^2$, hence the theorem follows.

Given any instance of D$k$S problem on graph $G=(V,~E)$, construct an instance
(denoted by SM-CA-I) of SM-CA problem as follows. It is defined on a one-way bipartite graph $G'=(V'=V_1\cup V_2,~E')$, where $V_1=V,~V_2=E$,
 the directed edge set $E'=\{(v,~e):\forall~v\in V_1,~e\in V_2$,
and $v$ is one of the endpoints of $e$ in $E\}$. The probability on each edge $e'=(v,~e)$ is $p_{ve}=1/2$.
 The target set $U=V_1\cup V_2$, for each node $e\in V_2$, $\tau_e=3/4$ and for each node $v\in V_1$, $\tau_v=1$.
 For any $k$, let $\eta=\eta(k)$ be the maximum threshold
requirement for which $\mathcal A$ outputs a solution for SM-CA with $k$ nodes.
 That is to say, $\mathcal A$ outputs a seed set with $k$ nodes
if the threshold is $\eta(k)$ and at least $k+1$ nodes if the threshold is larger than $\eta(k)$. \footnote{For any $k$, $\eta(k)$ can be computed efficiently by using algorithm $\mathcal{A}$ and linear search.}

It is clearly that, in SM-CA-I, nodes in $V_2$ are no better than nodes in $V_1$ as candidates of seed since
the target set is the set of all nodes, select a node in $V_2$ can only activate itself, but a node in $V_1$ may help to activate nodes in $V_2$. So here we assume that all seeds selected by algorithm $\mathcal{A}$ are from $V_1$.
Since for each edge $(v,e)\in E'$, $p_{ve}=1/2$ and for each node $e\in V_2$, $\tau_e=3/4$, an easy probability calculation implies that a node $e\in V_2$ can be cumulatively activated if and only if both endpoints of $e$ are selected as seeds.

Suppose the seed set of SM-CA-I with parameter $\eta=\eta(k)$ computed by algorithm $\mathcal A$ is $S'$,
 then we can use the corresponding node set $S$ in graph $G$
as an approximate solution of the D$k$S problem. Indeed, we have $|S|= k$. Since in SM-CA-I $S'$ cumulatively activates at least $\eta$ nodes, only $k$ of them are in $V_1$, so at least $\eta-k$ nodes are cumulatively activated in $V_2$. Therefore, in graph $G$ the number of edges induced by $S$ is at least $\eta-k$.

Without loss of generality, we can assume $\eta\geq k+\lfloor k/2\rfloor$, this is because we can
easily choose $k$ nodes from $V_1$ to cumulatively active $\lfloor k/2\rfloor$ nodes in $V_2$.
It is easy to check that $\eta-k\geq\frac{1}{3}(\eta-2)$.

Suppose the optimal solution of D$k$S problem contains $opt$ edges,
 then it is sufficient to show $opt\leq 2(\eta-2) r^2$. Indeed, if we can prove $opt\leq 2(\eta-2) r^2$, then we have $opt\leq 6(\eta-k) r^2$,
which means there is a $6r^2$-approximate algorithm for the D$k$S problem.

In SM-CA-I, based on the choice of $\eta$ and the fact that $\mathcal A$ is a $r$-approximate algorithm,
 any seed set with size $\lfloor k/r \rfloor$ can cumulatively activate
at most $\eta$ nodes. Thus, at most $\eta-\lfloor k/r\rfloor$ nodes in $V_2$ can be cumulatively activated
by any $\lfloor k/r \rfloor$ seeds in $V_1$.
 This is equivalent to the fact that
there are at most $\eta-\lfloor k/r \rfloor$ edges induced by any $\lfloor k/r \rfloor$ vertexes in $G$.
 Thus, for any $T\subseteq V$
with $|T|=k$, all possible subset of $\binom{k}{\lfloor k/r\rfloor}$ vertexes in $T$ can induce at most
$(\eta-\lfloor k/r\rfloor) \binom{k}{\lfloor k/r\rfloor}$ edges and each edge is counted
exactly $\binom{k-2}{\lfloor k/r\rfloor-2}$ times.
 So, if $k>2r$, the total number of edges induced by $T$ is at most
\begin{center}
$\frac{(\eta-\lfloor k/r\rfloor) \binom{k}{\lfloor k/r \rfloor}}{\binom{k-2}{\lfloor k/r \rfloor-2}}\leq
r^2(\eta-\lfloor k/r\rfloor) \frac{k-1}{k-r}<2(\eta-2)r^2$.
\end{center}
if $k\leq 2r$, then $opt\leq \binom{k}{2}\leq \frac{k^2}{2}\leq 2r^2$. By the arbitrary chosen of $T$, we have $opt\leq 2(\eta-2) r^2$ and this completes the proof.
\end{proof}

We remark that when $\eta=n$, it corresponds to the case of $k=n$ in the DkS problem, which has a trivial solution and makes
	the theorem statement vacuously true.
Thus we add $\eta < n$ just to emphasize that the theorem is only useful when $\eta < n$.

}

\subsection{Influence maximization with cumulative activation (IM-CA) problem}\label{sec:im-ca}

In IM-CA problem, we prove a strong inapproximability result even when the base graph is a bipartite graph.

\begin{theorem}\label{thm:maxhardness}
For any $\varepsilon>0$, it is NP-hard to approximate IM-CA problem within a factor of $N^{1-\varepsilon}$, where $N$ is the input size.
\end{theorem}
\OnlyInShort{
The proof of Theorem \ref{thm:maxhardness} is shown in our full version.
}
\OnlyInFull{
\begin{proof}
Similar to the proof of inapproximability result in \cite{kempe2003maximizing},
 We construct a reduction from SET COVER problem.
 The input of the SET COVER problem includes a ground set $W=\{w_1,~w_2,~\ldots,~w_n\}$, a collection of
subsets $S_1,S_2,\ldots,S_m\subset W$, and a positive integer $k<m$.
 The question is whether there exists $k$ subsets whose union is $W$.

Given an instance of the set cover problem, we construct an instance of IM-CA problem as follows:
There are three types of nodes, {\em set} nodes, {\em element} nodes, and {\em dummy} nodes.  There is a set node $u$ corresponding to each set, an element node $v$ corresponding to each element,
and a directed edge $(u,~v)$ with activation probability $p_{uv}=1$ if the element represented
by $v$ is belong to the set represented by $u$ and $p_{uv}=0$ otherwise. There are $n^c$
dummy nodes $x_1, x_2, \cdots, x_{n^c}$ (where $c=2\lceil {1\over \epsilon}\rceil+\lceil {\log m \over \log n} \rceil+1$), and there is a directed edge $(v, x)$
for each $v$ and $x$. The activation probability on $(v, x)$ is $p_{vx}=1/2$.
 The activation thresholds of set nodes, element nodes and dummy nodes are $\tau_u=\tau_v=1,~\tau_x=1-\frac{1}{2^n}$, respectively.
The budget of the size of a seed set is $k$ and the target set is all nodes. Notice that the input size of our IM-CA problem is $N=n^c+n+m$, so $N^{1-\epsilon}<\frac{2n^c}{N^\epsilon}\leq \frac{n^c}{n+k}$.

Under our construction, if there exists a collection of $k$ sets covering all elements in $W$ for SET COVER problem,
 then in IM-CA problem, the node set corresponding to the collection denoted by $C$
 will cumulatively activate all element nodes and all dummy nodes.
 In total, there will be $n^c+n+k$ nodes become cumulatively active.
On the other hand, let's consider the case if there is no set cover of size $k$.
Again we can assume all the seeds are selected from set nodes, since as a candidate for seeds, set nodes are more efficient than element nodes and dummy nodes. 
 Thus, if there is no set cover of size $k$, then we cannot find $k$ seeds which activate all the element nodes, hence none of the dummy notes are activated. Therefore, the total number of nodes cumulatively activated are no more than $n+k$.
 It follows that if a polynomial algorithm can approximate IM-CA problem within $N^{1-\epsilon}$, 
 then we can answer the decision problem of the SET COVER problem in polynomial time, this is impossible under the assumption P$\neq$ NP.
\end{proof}
}
\section{\leftline {Efficient Heuristic Algorithms}}\label{sec:heuristicalgorithms}
In Section \ref{sec:algorithmsandhardness}, we prove that both SM-CA with $\eta<n$ and IM-CA are hard to approximate.
Despite this difficulty, in this section we present efficient heuristic algorithms based on the greedy strategy,
	in order to tackle the problem
	in practice.
 We first show the outline of our greedy strategies in Section \ref{sec:greedy algorithm based on MC}.
  In Section \ref{sec:improved gerrdy algorithms based on RRS}, we adopt an efficient method to design scalability greedy algorithms.

\subsection{\leftline {Greedy Strategies}}\label{sec:greedy algorithm based on MC}
In this section, we introduce two possible greedy strategies for SM-CA problem and IM-CA problem.

From Section \ref{sec:eta=n}, we know that greedy by the surrogate function $f(S)=\sum_{u\in V} \min\{P_u(S), \tau_u\}$ can guarantee good approximation ratio for SM-CA problem with $\eta=n$.
Thus, intuitively we could adopt $f(S)$ as our surrogate objective even when $\eta < n$ and apply the greedy strategy based on $f(S)$.
However, our initial experiments demonstrate that directly adopting $f(S)$ is less effective, especially when seed set size is relative small.
We believe that this is because greedy on $f(S)$ would prefer larger increment of $P_u(S)$ far below $\tau_u$ over smaller increment
	of $P_u(S)$ close to $\tau_u$, but the latter actually provides new cumulative activations.
To guide seed selection towards the latter case, we generalize $f(S)$ to $F(S)=\sum_{u\in V} \min\{P_u(S), c\tau_u\}$ by introducing an additional parameter $c$.

A large $c$ reduces the difficulty of lifting $P_u(S)$ over the threshold $\tau_u$ when it is getting close to $\tau_u$, but
	it continuously rewards $P_u(S)$ above $\tau_u$, while a $c$ close to $1$ has the reverse effect.
Essentially $c$ balances between the truncated surrogate $f(S)$ (when $c=1$) and the expected influence function $\sigma(S)$ (when
	$c$ is large).
Thus, our first greedy strategy is to use $F(S)$ with a proper tuned $c$ as the greedy objective,
	and we call it the {\em balanced truncation greedy (BTG)} strategy.

The second strategy is to apply greedy on the objective function $\rho(S)$ directly.
That is, we select the node with the largest increment to $\rho(S)$ in each step.
However, since $\rho(S)$ is a discrete rounding function, there could be many nodes having the same effect (or no effect at all)
	in any step.
For tie-breaking, we select nodes according to $f(S)$, which is equal to $\sigma(S)$ under this situation.
In summary, the second strategy preferentially selects nodes promoting $\rho(S)$ most, then chooses the node contributing to $f(S)$ most
	among nodes having the same promotion to $\rho(S)$.
In this strategy, the objective function (i.e. $\rho(S)$) plays a dominant role in selecting seeds. We call it the
	{\em activation dominance greedy (ADG)} strategy.



During the process of greedy algorithms, we need to estimate $P_u(S)$ for each node $u\in V$.
 It will be very expensive if we do this estimation by Monte Carlo simulations.
 Specially, by Lemma \ref{lem:error estimate}, we need to simulate $O(n^2\ln n)$ times to guarantee the accuracy,
  each simulation takes $O(m)$ time in the worst case.
 Thus, it takes $O(n^2m\ln n)$ for each node $u$ to estimate $P_u(S)$.
To improve the efficiency, we adopt a new approach named \emph{reverse reachable ret (RR set)},
	as we describe in the next section.


\subsection{Greedy Algorithms Based on RR Set}\label{sec:improved gerrdy algorithms based on RRS}
In this section, we present our efficient algorithms based on RR set.
We first introduce the background of RR set.
 RR set was first proposed by Borgs et al. in 2014 \cite{borgs2014rrset}
	 to provide the first near-linear-time algorithm for the
	 classical influence maximization problem in \cite{kempe2003maximizing}.
The approach is further optimized later in a series of follow-up work \cite{tang2015rrset,tang2014newrrset,mtai2016sigmod}.
  The definition of RR set is as follows:
\begin{definition}[Reverse Reachable Set]
Let $u$ be a node in $G$, and $g$ be a random graph obtained by independently removing each edge $e=(v,w)$ in $G$
with probability $1-p_e$. The {\em reverse reachable set (RR set)} for $u$ is the set of nodes in $g$ that can reach $u$.
\end{definition}

 Borgs et al. established a crucial connection between RR set and the influence propagation process on $G$. We restatement it in Lemma \ref{lem:$p(active)=p(overlap)$}.

\begin{lemma}[\cite{borgs2014rrset}] \label{lem:$p(active)=p(overlap)$}
Let $S$ be a seed set and $u$ be a fixed node. Suppose $R_u$ is an RR set for $u$ generated from $G$,
 then $P_u(S)$ equals the probability that $S$ overlaps with $R_u$, that is,
\begin{center}
$
P_u(S)=\Pr(S\cap R_u\neq \emptyset)$.
\end{center}
\end{lemma}

Now we introduce our new method to estimate $P_u(S)$ for each node $u\in V$.
We first generate $\theta$ RR sets for $u$ independently. Let $\mathcal R_u$ be the collection of all generated RR sets for $u$.
 For any node set $S$, let $\mathcal {F}_{\mathcal R_u}(S)$ be the fraction of RR sets in ${\mathcal R_u}$ overlapping with $S$. That is, $\mathcal {F}_{\mathcal R_u}(S)\triangleq |\{R_u\in \mathcal{R}_u: R_u\cap S\neq \emptyset\}|/\theta.$
Then for any $u\in V$, we use $\mathcal {F}_{\mathcal R_u}(S)$ as the estimation of $P_u(S)$.
 We can prove that we can bound the estimation error if $\theta$ is large enough.
 \OnlyInShort{
 The proof of Lemma \ref{lem:rrsetestimate} is also shown in our full version \cite{fullversion}.
}
\begin{lemma}\label{lem:rrsetestimate}
For any $\varepsilon>0$, if $\theta$ satisfies $\theta\geq \ln (2n)/2\varepsilon^2$, then for each node $u\in V$,
\begin{center}
$\Pr[|\mathcal {F}_{\mathcal R_u}(S)-P_u(S)|\geq \varepsilon]\leq n^{-1}$.
\end{center}
\end{lemma}
\OnlyInFull{
\begin{proof}
 Let $X\triangleq\theta\mathcal {F}_{\mathcal R_u}(S)$, then $X$ is the number of RR sets in $\mathcal R_u$ overlapping with $S$.
 Moreover, $X$ can be regarded as the sum of $\theta$ $i.i.d.$ Bernoulli variables.
 Specifically, let $X=\sum_{i=1}^\theta X_i$ where $X_i=1$ if $S$ overlaps with the $i$-th RR set in $\mathcal R_u$ and $X_i=0$ otherwise.
 Based on Lemma \ref{lem:$p(active)=p(overlap)$}, we have $E[X]=\sum_{i=1}^\theta E[X_i]=\theta P_u(S)$.
By the Hoeffding's inequality and the condition $\theta\geq \ln (2n)/2\varepsilon^2$, we have
\begin{equation*}
\begin{aligned}
&\Pr(|\mathcal {F}_{\mathcal R_u}(S)-P_u(S)|\geq \varepsilon)=\Pr(|\theta\mathcal {F}_{\mathcal R_u}(S)-\theta P_u(S)|\geq \theta\varepsilon)\\
&=\Pr(|X-E[X]|\geq \theta\varepsilon)\leq 2\exp(-2(\theta\varepsilon)^2/\theta)\leq n^{-1}.
\end{aligned}
\end{equation*}
\end{proof}
}
We now present our greedy algorithms.
Recall that we use two greedy functions:
$F(S)=\sum_{u\in V}\min\{P_u(S),c\tau_u\}$ and
	$\rho(S)= \sum_{u\in V}\I\{P_u(S) \ge \tau_u \}$, where $\I$ is the indicator function.

In order to make it easier to understand, we describe the processes of selecting seeds in subprograms.
 We first present the framework of the whole greedy algorithm for IM-CA problem in Algorithm \ref{alg:Greedy Framework of IM-CA problem}.

  \begin{algorithm}[t]
\renewcommand{\algorithmicrequire}{\textbf{Input:}}
\renewcommand\algorithmicensure {\textbf{Output:}}
    \caption{Framework of greedy algorithm for IM-CA problem} \label{alg:Greedy Framework of IM-CA problem}
    \begin{algorithmic}[1]
    \REQUIRE ~~ $G=(V,~ E),~ \{p_{uv}\}_{(u,~v)\in E}, \{\tau_u\}_{u\in V},~ k, ~\theta$\\
    \ENSURE ~~  Seed set $S$\\
\STATE  set $S=\emptyset$ \label{algline:initialseedset}
\STATE  generate $\theta$ RR sets for each node $u\in V$: $\{\mathcal R_u\}_{u\in V}$ \label{algline:generaterrset}
\STATE  set $req(u)=\tau_u\theta $ for each node $u\in V$ \label{algline:initialreq}
\FOR  {$j=1$ to $k$} \label{algline:beginselect}
    \STATE  $x=$ SS($G, \{p_{uv}\}_{(u,v)\in E}, \{req(u)\}_{u\in V}, ~\{\mathcal R_u\}_{u\in V}$)\label{algline:callsubprogram}
    \STATE  /*SS is a general term of SSBT and SSAD*/
    \STATE  $S=S\cup \{x\}$ \label{algline:selectaseed}
    \STATE  remove all $RR$ Sets containing $x$ \label{algline:removerrset}
    \FOR  {each $u$ in $V$} \label{algline:beginupdatereq}
        \STATE  $rem(u)$: the number of $RR$ Sets removed from $\mathcal R_u$
        \STATE  $req(u)=req(u)-rem(u)$
    \ENDFOR\label{algline:endupdatereq}
\ENDFOR
\RETURN $S$
\end{algorithmic}
\end{algorithm}

 In Algorithm \ref{alg:Greedy Framework of IM-CA problem}, we first initialize for the seed set $S$ (line \ref{algline:initialseedset}).
     Then we generate $\theta$ RR sets for each node $u$ in $V$.
     Let $\mathcal R_u$ be the collection of RR sets for $u$.
     In line \ref{algline:initialreq}, $req(u)$ is the requirement of node $u\in V$, which
     is the number of RR sets in $\mathcal R_u$ that needs to be hit by a seed set so that $u$ can become
     cumulatively active. We say a set $S$ hits an RR set $R$ if $S\cap R\neq \emptyset$.
     Based on Lemma \ref{lem:rrsetestimate}, $u$ is cumulatively active only if there are at least $\theta \tau_u$ RR sets in $\mathcal R_u$ hit by the seed set.
     Thus, we set $req(u)=\tau_u\theta $ for each node in $u\in V$.

 At each step, we add a new node $x$ into the current seed set (line \ref{algline:callsubprogram}).
    After $x$ is selected, we need to remove all RR sets containing $x$ and update the requirements for all nodes.
    The algorithm ends when $|S|=k$.

  Note that Algorithm \ref{alg:Greedy Framework of IM-CA problem} needs to call the seed selecting procedures (line \ref{algline:callsubprogram}). Here, SS($\cdot$) is a general term for our two subprograms \selectbyf~(Selecting Seeds via Balanced Truncation strategy) and \selectbyrhof~(Selecting Seeds via Activation Dominance strategy).
Specifically, \selectbyf ~(Procedure \ref{alg:\selectbyf}) is the subprogram that selects one node with the largest marginal increment to $F(S)$ into the current seed set $S$.
\selectbyrhof ~(Procedure \ref{alg:SelectByrho}) is the subprogram selecting the node
	with the largest marginal increment to $\rho(S)$, with tie-breaking on $f(S)$.
 The algorithm calling \selectbyf~is named as \FMAX (Balanced Truncation Greedy algorithm for IM-CA problem) and the algorithm calling \selectbyrhof~is named as \RHOMAX (Activation Dominance Greedy algorithm for IM-CA problem ).

\begin{algorithm}[t]
\floatname{algorithm}{Procedure}
\renewcommand{\algorithmicrequire}{\textbf{Input:}}
\renewcommand\algorithmicensure {\textbf{Output:}}
    \caption{\selectbyf: Selecting Seeds via Balanced Truncation strategy} \label{alg:\selectbyf}
    \begin{algorithmic}[1]
    \REQUIRE  $G=(V, E),~\{p_{uv}\}_{(u, v)\in E},~\{req(u)\}_{u\in V},~\{\mathcal R_u\}_{u\in V}$\\
    \ENSURE   a new seed \\
    \STATE  set $inc(v)=0$ for all $v\in V$
    \FOR {each node $u\in V$ and $req(u)>0$} \label{sub1line:begincomputeinc}
        \FOR  {each node $v\in \bigcup\mathcal R_u$}
        \STATE  /*compute the marginal increment of $v$*/
            \STATE  $inc(v)=inc(v)+ \min\{overlap(v,~\mathcal R_u),~c\cdot req(u)\}$
        \ENDFOR
    \ENDFOR \label{sub1line:endcomputeinc}
    \STATE  select $x=\arg\max_v inc(v)$
    \RETURN  $x$
\end{algorithmic}
\end{algorithm}

Now we describe our two subprograms SSBT and SSAD.
    We first introduce SSBT in Procedure \ref{alg:\selectbyf}.
    Let $inc(v)$ be the value of the marginal increment generated by any node $v\in V$,  $overlap(v,\mathcal R_u)$ be the number
    of RR sets in $\mathcal R_u$ overlapping with node $v$.
    In the main loop of \selectbyf, we select the node providing the largest marginal increment to $F(S)$.
    To this end, for each node $v\in V$, we compute the marginal increment of $v$ to all nodes which are not cumulatively active yet.
    Based on Lemma \ref{lem:$p(active)=p(overlap)$},
    the marginal increment of a node $v$ to node $u$ can be measured by $\min\{overlap(v,~\mathcal R_u),~c\cdot req(u)\}$.
    Summing up the increments of $v$ on all not-yet cumulatively active nodes, we can obtain $inc(v)$ (see details in lines \ref{sub1line:begincomputeinc} - \ref{sub1line:endcomputeinc}).
    Then, we choose the node with the maximum $inc(v)$.

\begin{algorithm}[t]
\floatname{algorithm}{Procedure}
\renewcommand{\algorithmicrequire}{\textbf{Input:}}
\renewcommand\algorithmicensure {\textbf{Output:}}
    \caption{\selectbyrhof: Selecting Seeds via Activation Dominance strategy} \label{alg:SelectByrho}
    \begin{algorithmic}[1]
    \REQUIRE   $G=(V,E),~ \{p_{uv}\}_{(u,v)\in E},~ \{req(u)\}_{u\in V},~ \{\mathcal R_u\}_{u\in V}$\\
    \ENSURE   a new seed\\
    \STATE  set $inc(v)=0$ for all $v\in V$
    \FOR {each node $u\in V$ and $req(u)>0$} \label{sub2line:begincomputerhoinc}
        \FOR  {$v\in \bigcup\mathcal R_u$}
        \STATE  /*compute the marginal increment of $v$*/
            \STATE  $inc(v)=inc(v)+ \I\{overlap(v,~\mathcal R_u) \ge req(u)\}$
        \ENDFOR
    \ENDFOR \label{sub2line:endcomputerhoinc}
    \STATE  /*select one better node from nodes with the largest marginal increment*/
    \FOR {each node $u\in V$ and $req(u)>0$} \label{sub2line:begincomputefinc}
        \FOR  {$v\in \bigcup\mathcal R_u$ with the largest $inc(v)$ values}
            \STATE  $inc(v)=inc(v)+ \min\{overlap(v,~\mathcal R_u),~req(u)\}$
        \ENDFOR
    \ENDFOR \label{sub2line:endcomputefinc}
    \STATE  select $x=\arg \max_v inc(v)$
    \RETURN  $x$
\end{algorithmic}
\end{algorithm}

Another greedy strategy is shown in \selectbyrhof~(Procedure \ref{alg:SelectByrho}).
    In this procedure, we first find nodes with the largest marginal increment to $\rho(S)$.
    For any node $v$,  the marginal increasing $inc(v)$ can be denoted by:
    \begin{center}
     $\sum_{u:~req(u)>0}\I\{overlap(v, ~\mathcal R_u) \ge req(u) \}$
    \end{center}
   There may be many nodes with the same value of $inc(v)$ due to the truncation operation of $\rho(S)$.
   To break the tie, we choose the node with the maximum marginal increase to $f(S)$ among all nodes with the largest $inc$ value (lines \ref{sub2line:begincomputefinc} - \ref{sub2line:endcomputefinc}).

The framework of the whole greedy algorithm for SM-CA problem follows the same structure with Algorithm \ref{alg:Greedy Framework of IM-CA problem}.
    The only difference between these two algorithms is the stopping condition.
    For SM-CA, the algorithm stops when the number of cumulatively active nodes is no less than $\eta$.
    The corresponding algorithms of SM-CA problem are named as \FMIN (Balanced Truncation Greedy algorithm for SM-CA problem) and \RHOMIN (Activation Dominance Greedy algorithm for SM-CA problem).
    Other details of the algorithms are essentially the same as the algorithms for IM-CA and are thus omitted.

\textbf{Time complexity.} Now we analyze the time complexity of \FMAX and ADG-IM-CA.
 Let $EPT$ be the expected sum of in-degrees of all nodes in a random RR set,
	 which is the same as the expected time of generating an RR set.
 Thus, the total expected time of the generation is $O(n\theta\cdot EPT)$. By Lemma~\ref{lem:rrsetestimate},  $\theta=\ln (2n)/2\varepsilon^2$ is enough for accuracy.
 Thus, the expected generation time is $O(n\ln n\cdot EPT/\varepsilon^2)$.
 Besides the generation time, the main time cost depends on SS($\cdot$) since other operations only take time $O(n)$.
 For each node $u\in V$, let $EPTV_u$ be the expected number of nodes in $\bigcup \mathcal R_u$ and $EPTV=\frac{1}{n}\sum_{u\in V}EPTV_u$.
 Then, both \selectbyf~and \selectbyrhof~takes time $O(n\cdot EPTV)$ in expectation.
 Hence, the expected time complexity of both \FMAX and \RHOMAX is  $O(n(k EPTV+\ln n\cdot EPT/\varepsilon^2))$.

 At each step of \FMIN and \RHOMIN, the number of cumulatively active nodes increases at least 1 since the selected seed contributes 1 to $\rho(\cdot)$, which means the times of the outer loop is at most $\eta$.
 Thus, the expected time cost of \FMIN and \RHOMIN is
 $O(n(\eta EPTV+\ln n\cdot EPT/\varepsilon^2))$.

\section{Experiments}\label{sec:experiments}

In order to test the performance of our heuristic algorithms, we conduct experiments on real social networks.
Our experiments are run on a machine with a 2.4GHz Intel(R) Xeon(R) E5-2670 CPU, 2 processors (16 cores), 64GB memory and Red Hat Enterprise Linux Server release 6.3 (64bit). All algorithms tested in this paper are written in C++ and compiled with g++ 4.8.4.

\subsection{Experiment setup}

\textbf{Datasets.} We use three real-world networks in our experiments: Flixster, NetPHY and DBLP. Table \ref{table:datasets} shows the datasets used in our experiments, in which AOD denotes the average out degree of a dataset.
\begin{table}[h]
\caption{Datasets}
\begin{tabular}{|c|c|c|c|c|}
  \hline
  {\small Name} & {\small \# Node} & {\small \#Edge} & {\small Type} & {\small AOD}\\ \hline
   Flixster & 29K & 174K & directed & 6.0\\ \hline
  NetPHY &  37K  & 348K & undirected & 18.8\\ \hline
  DBLP & 655k & 2M & undirected & 6.1\\
  \hline
\end{tabular}
\label{table:datasets}
\end{table}
\OnlyInFull{
The first network is Flixster, which is an American movie rating social site for discovering new movies.
In the Flixster graph, each node represents a user and a directed edge $e=(u,~ v)$ represents that
$u$ and $v$ rate the same movie and $v$ rates the movie shortly after $u$.
We simply use one specific topic in this network with 29357 nodes and 174939 directed edges.
And we learn the active probabilities on edges by using the Topic-aware Independent Cascade Model presented in \cite{barbieri2012topic}. The mean of edge probabilities is 0.118 and the standard deviation is 0.025.

The second one, called NetPHY,
is the same as the one used in \cite{chen2010sharpphard,goyal2012minimizing,simpath}.
It is an academic collaboration network extracted from the ``Physics" section from
arXiv(http:// www. arXiv.org). The nodes in NetPHY are authors and undirected edges represent coauthorship relations.
We use data from year 1991 to year 2003 which includes 37154 nodes and 348322 edges.
The influence probabilities on edges are assigned by weighted cascade model \cite{kempe2003maximizing}.
Specifically, for each edge $(u, v)\in E$, we set $p_{uv}=c(u,v)/d(v)$, where $d(v)$ is the number of published papers of author $v$ and $c(u,v)$ is the number of papers that $u$ and $v$ collaborated.
In this network, the mean of edge probabilities is $0.107$ and the standard deviation is $0.025$.

The last one is a larger collaboration network DBLP maintained by Michael Ley (654628 nodes and 2056186 edges).
The method of generating edge probability is the same as that in NetPHY.
We follow the TRIVALENCY model \cite{chen2012scalable} to assign edge probabilities:
On every edge $e=(u, v)$, we uniformly select a probability from the set $\{0.1, 0.01, 0.001\}$ at random,
which corresponds to high, medium and low influences.
Under the above method, the mean of edge probabilities is 0.069 and the standard deviation is 0.002.
}

\textbf{Algorithms.}
We test our algorithms using both BTG and ATG strategies. For comparison, we use the following baseline algorithms.
%
%
%
\begin{itemize}
\item {\sf TIM$^+$:} {\sf TIM$^+$} is a greedy algorithm presented in \cite{tang2014newrrset}. The basic greedy rule is to choose the node covers the maximum number of Random RR sets (see more details in \cite{tang2014newrrset}).
\item {\sf High-degree:} {\sf High-degree} generates seed set sequence by the decreasing order of the out-degree of nodes. It is popular to consider high degree nodes as influential nodes in social and other networks.
\item {\sf PageRank:}
It is the popular algorithm used for ranking web pages \cite{pagerank}. The transition probability on edge
$e=(v,u)$ is $p_{uv}/\sum_{w:(w,v)\in E}p_{wv}$. In the {\sf PageRank} algorithm, higher $p_{uv}$ indicates $u$ is more influential to $v$
and thus $v$ should vote $u$ higher.
This is the reason we relate $p_{uv}$ with the transition probability of the reverse link $(v,u)$, which is the same as in earlier studies such as~\cite{chen2012scalable}.
We use 0.15 as the restart probability and use the power method to compute the PageRank values. The stopping criteria of computing PageRank values is when two consecutive iterations are different for at most $10^{-4}$ in $L_1$ norm. We select seeds by decreasing order of the PageRank values.
\item {\sf Random:} As a trivial baseline, {\sf Random} selects seeds sequence in random order.
\end{itemize}

For all the above algorithms, we use the same Monte Carlo method to compute the number of cumulative active nodes. The stopping criteria of IM-CA problem and SM-CA problem are the number of seeds is $k$ and the number of cumulative active nodes exceeds $\eta$, respectively.

\textbf{Parameters.} In our experiment, we take two simplified operations: First, we set parameters $\theta=1000$, $\varepsilon=0.1$ in Flixster and NetPHY, this setting is enough for the requirement in Lemma \ref{lem:rrsetestimate}. In DBLP, we set $\theta=500$, in practice, this is good enough for illustrating our results.
We also suppose that all nodes have the same activation threshold $\tau$.
In our experiments, $\tau$ ranges from $0.1$ to $0.9$.

\begin{figure*}[t]
\centering
	\includegraphics[width=0.4\textwidth]{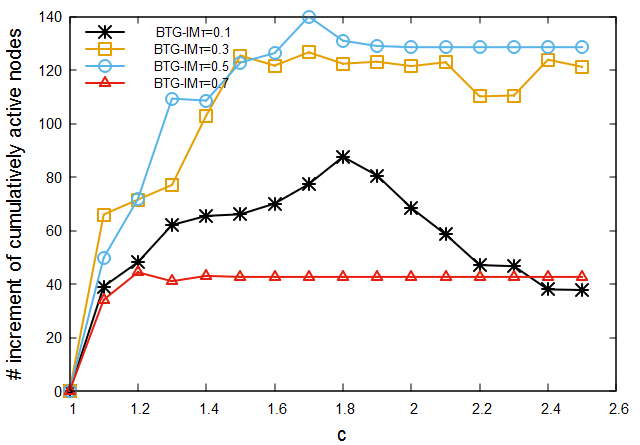}
   \caption{Results of $c$ on Flixster}	
\label{fig:flixster_c_max}
\end{figure*}

In order to determine a proper value of parameter $c$ for the BTG strategy,
	we fix $k=500$ on Flixster, then we implement \FMAX on different $\tau$ and $c$.
 Figure \ref{fig:flixster_c_max} shows the result. In order to present more clearly, we set the ordinate as the number of cumulatively active nodes minus the corresponding number under $c=1$.
 In general, we can see that $\tau$ greater than $1$ yields better result, and $\tau$ between $1.6$ and $1.8$ provides close-to-best results in all cases.
   In the rest, we choose $c=1.7$ for all tests (except tests on SM-CA, which involves a large
   number of seeds).
	Tests on other datasets yield similar results.


\subsection{Experiment results}

\textbf{Experiment results on spreading performance of IM-CA problem.}
\OnlyInFull{
Figure~\ref{fig:flixstermax}, Figure \ref{fig:NetPHYmax} and Figure \ref{fig:DBLPmax} show the comparison of different algorithms on Flixster, NetPHY and DBLP respectively.
}
\OnlyInShort{
Figure \ref{fig:Results of IM-CA on different datasets} shows the comparison of different algorithms on Flixster, NetPHY and DBLP.
}
These figures reflect the spreading performances with different $\tau$, varying $k$ from 1 to 500.

\OnlyInFull{
\begin{figure*}[ht!]
	\centering
	\subfigure[Flixster $\tau=0.1$]
	{
		\label{fig:flixster_0.1_max}
		\includegraphics[width=0.3\textwidth]{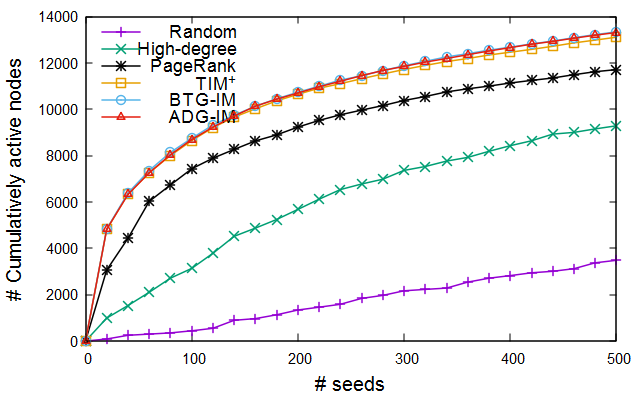}
	}
    \subfigure[Flixster $\tau=0.3$]
	{
		\label{fig:flixster_0.3_max}
		\includegraphics[width=0.3\textwidth]{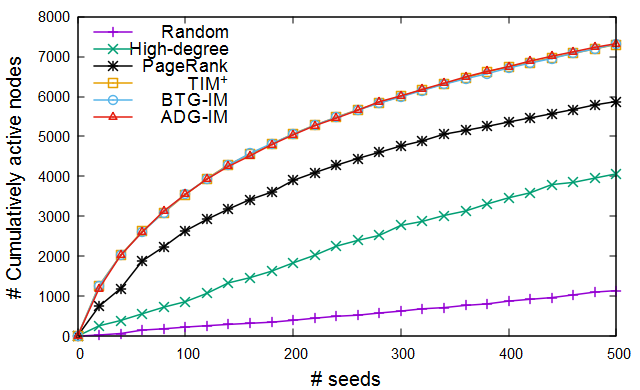}
	}
    \subfigure[Flixster $\tau=0.5$]
	{
		\label{fig:flixster_0.5_max}
		\includegraphics[width=0.3\textwidth]{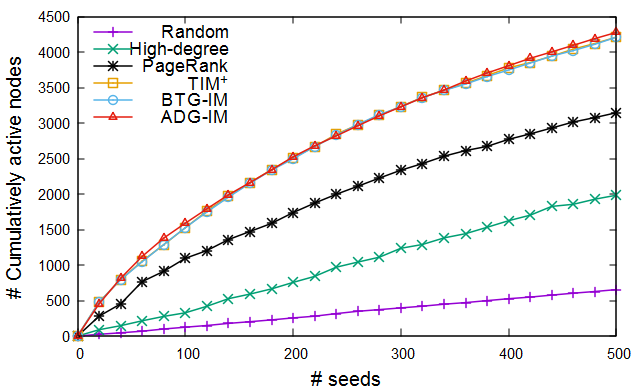}
	}
    \subfigure[Flixster $\tau=0.7$]
	{
		\label{fig:flixster_0.7_max}
		\includegraphics[width=0.3\textwidth]{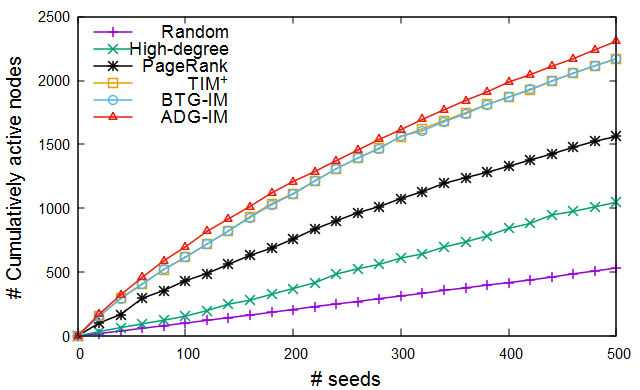}
	}
    \subfigure[Flixster $\tau=0.9$]
	{
		\label{fig:flixster_0.9_max}
		\includegraphics[width=0.25\textwidth]{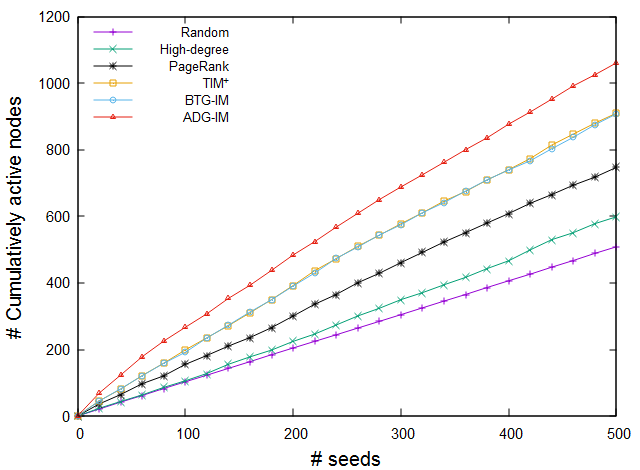}
	}
	\caption{Results of IM-CA on Flixster}
	\label{fig:flixstermax}
\end{figure*}

\begin{figure*}[!ht]
	\centering
	\subfigure[NetPHY $\tau=0.1$]
	{
		\label{fig:NetPHY_0.1_max}
		\includegraphics[width=0.3\textwidth]{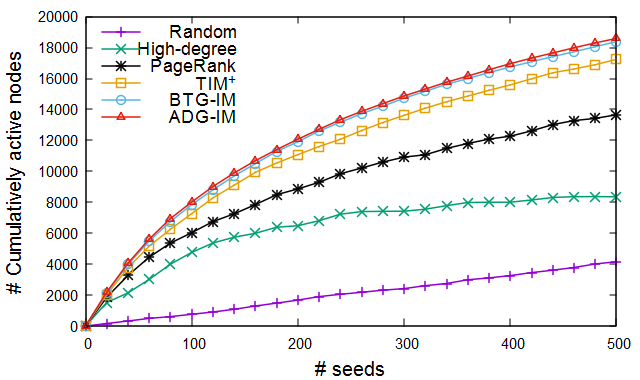}
	}
    \subfigure[NetPHY $\tau=0.3$]
	{
		\label{fig:NetPHY_0.3_max}
		\includegraphics[width=0.3\textwidth]{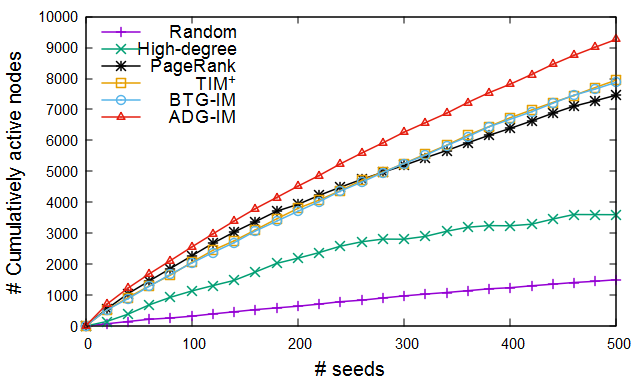}
	}
    \subfigure[NetPHY $\tau=0.5$]
	{
		\label{fig:NetPHY_0.5_max}
		\includegraphics[width=0.3\textwidth]{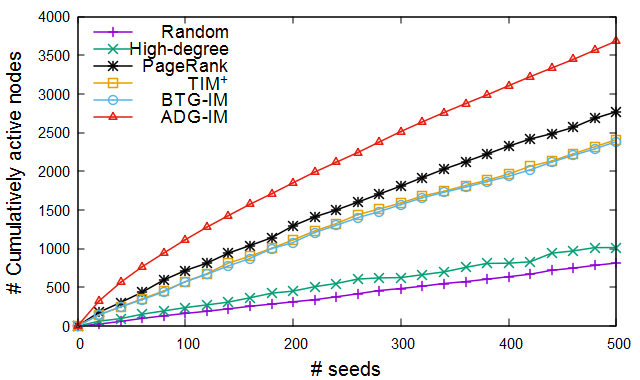}
	}
    \subfigure[NetPHY $\tau=0.7$]
	{
		\label{fig:NetPHY_0.7_max}
		\includegraphics[width=0.3\textwidth]{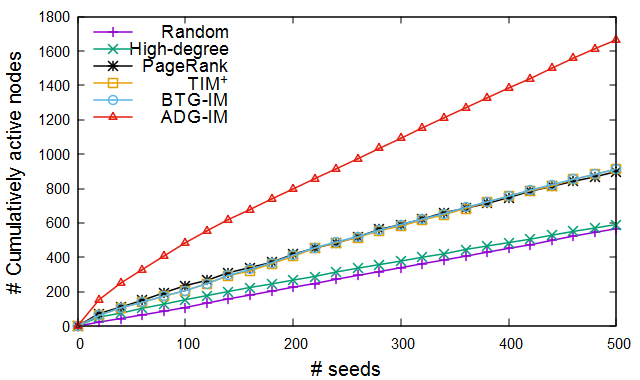}
	}
   \subfigure[NetPHY $\tau=0.9$]
	{
		\label{fig:NetPHY_0.9_max}
		\includegraphics[width=0.28\textwidth]{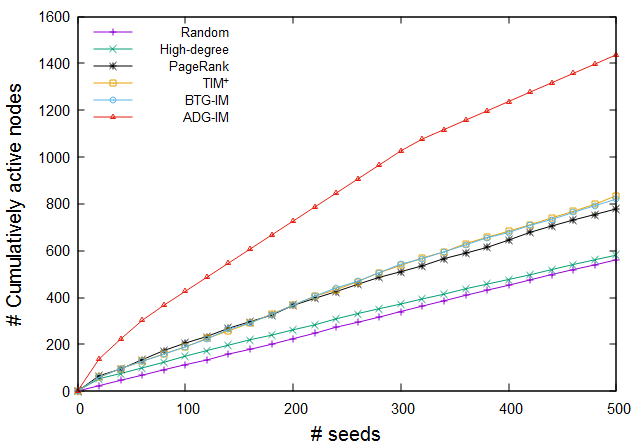}
	}
	\caption{Results of IM-CA on NetPHY}
	\label{fig:NetPHYmax}
\end{figure*}

\begin{figure*}[!ht]
	\centering
	\subfigure[DBLP $\tau=0.1$]
	{
		\label{fig:DBLP_0.1_max}
		\includegraphics[width=0.3\textwidth]{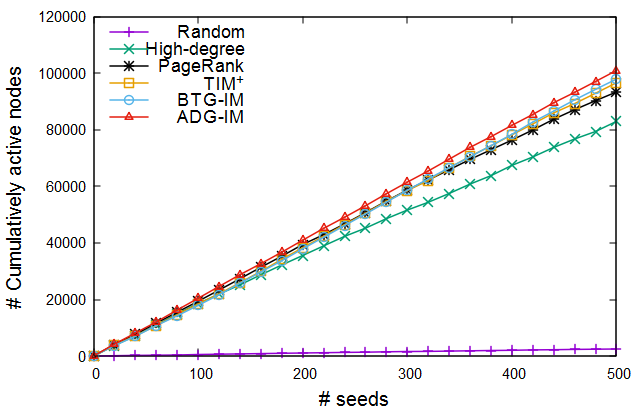}
	}
    \subfigure[DBLP $\tau=0.3$]
	{
		\label{fig:DBLP_0.3_max}
		\includegraphics[width=0.3\textwidth]{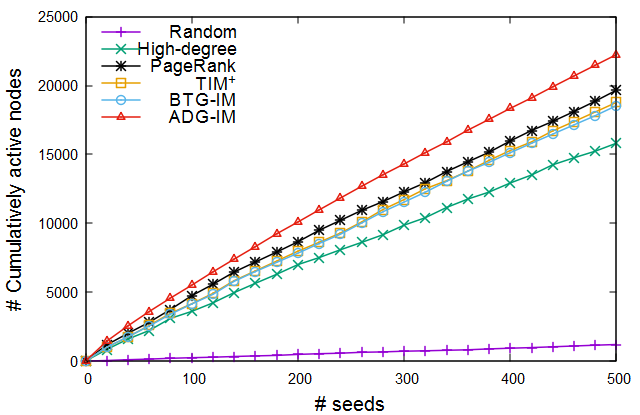}
	}
    \subfigure[DBLP $\tau=0.5$]
	{
		\label{fig:DBLP_0.5_max}
		\includegraphics[width=0.3\textwidth]{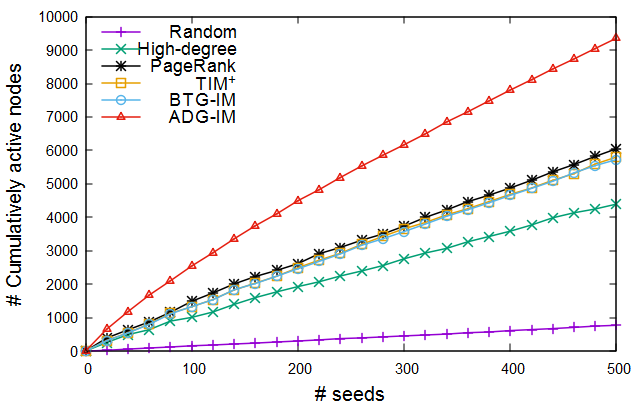}
	}
    \subfigure[DBLP $\tau=0.7$]
	{
		\label{fig:DBLP_0.7_max}
		\includegraphics[width=0.3\textwidth]{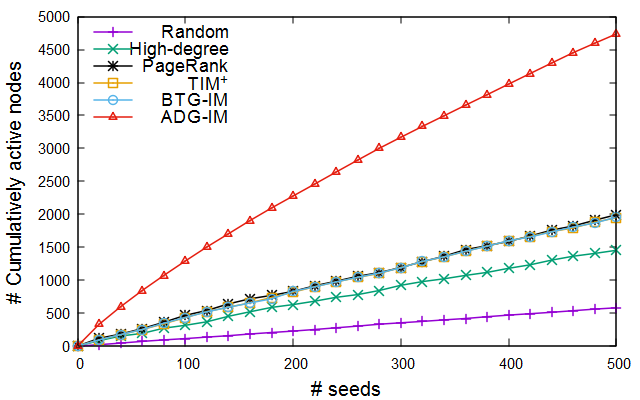}
	}
    \subfigure[DBLP $\tau=0.9$]
	{
		\label{fig:DBLP_0.9_max}
		\includegraphics[width=0.25\textwidth]{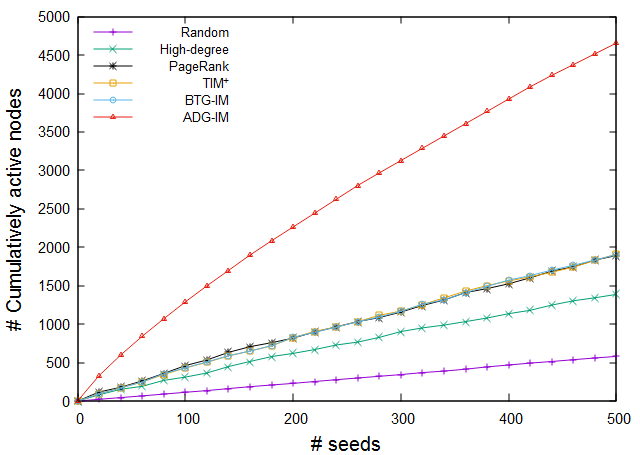}
	}
	\caption{Results of IM-CA on DBLP}
	\label{fig:DBLPmax}
\end{figure*}
}
Figures on these three datasets reflect some common features. Firstly, the performances of {\sf \FMAX} and {\sf TIM$^+$} are similar when $\tau$ is large. This is because when $c=1.7$, $c\tau>1$ if $\tau>0.59$. In this case, {\sf \FMAX} is exactly the same as {\sf TIM$^+$}. While  $\tau$ is smaller, {\sf \FMAX} is better than {\sf TIM$^+$} sometimes. For example, on NetPHY with $\tau=0.1$ and $k=500$ (Figure \ref{fig:NetPHY_0.1_max}) {\sf \FMAX} cumulatively activates nodes with size 7.6\% more than those cumulatively activated by {\sf TIM$^+$}.

Secondly, {\sf \RHOMAX} performs similar to {\sf TIM$^+$} and {\sf \FMAX} when $\tau$ is small. However, the curves of {\sf \RHOMAX} and {\sf TIM$^+$} become separated with the increase of $\tau$ on all datasets. {\sf \RHOMAX} outperforms all other algorithms significantly when $\tau$ is large.
In particular, on NetPHY with $\tau=0.7$ and $k=500$ (Figure \ref{fig:NetPHY_0.7_max}), the number of nodes cumulatively activated by {\sf \RHOMAX} is 90.8\% more than that by {\sf TIM$^+$}, 89.5\% more than {\sf \FMAX}, 88.8\% more than {\sf PageRank}, 190.2\% more than {\sf High-degree}, 227.0\% more than {\sf Random}.
For the two million edges DBLP dataset, when $\tau=0.7$ and $k=500$ (Figure \ref{fig:DBLP_0.7_max}), {\sf \RHOMAX} cumulatively activates nodes with size 162.0\% more than those cumulatively activated by {\sf TIM$^+$}, 162.0\% more than {\sf \FMAX}, 157.0\% more than {\sf PageRank}, 246.7\% more than {\sf High-degree}, 801.1\% more than {\sf Random}.
We think this feature is mainly because the cumulative activation is easy when $\tau$ is small, thus, the seed set generated by {\sf TIM$^+$} is likely to cumulatively activate enough nodes.
However, when $\tau$ is large, most nodes are not easy to be cumulatively activated. In this case, selecting seeds directly contribute to $\rho(S)$ may be more effective.
\OnlyInFull{
\textbf{Experiment results on the activation threshold $\tau$.}
To see the change of cumulative activation influence size with the increase of parameter $\tau$, we also conduct experiments on different $\tau$ on NetPHY (Figure~\ref{fig:phy_tau}).
\begin{figure}[h]
	\centering
	\subfigure[NetPHY $k=240$]
	{
		\label{fig:phy_tau_240}
		\includegraphics[width=0.3\textwidth]{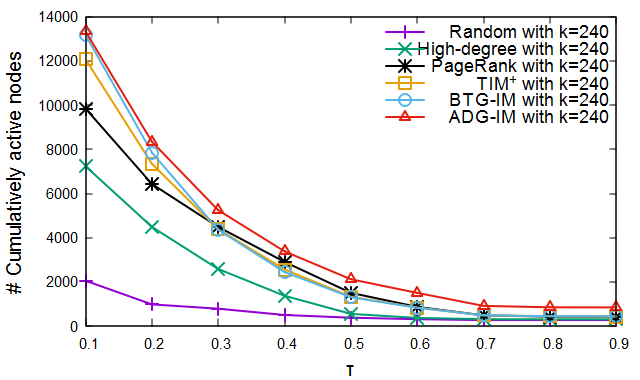}
	}
	\subfigure[NetPHY $k=500$]
	{
		\label{fig:phy_tau_500}
		\includegraphics[width=0.3\textwidth]{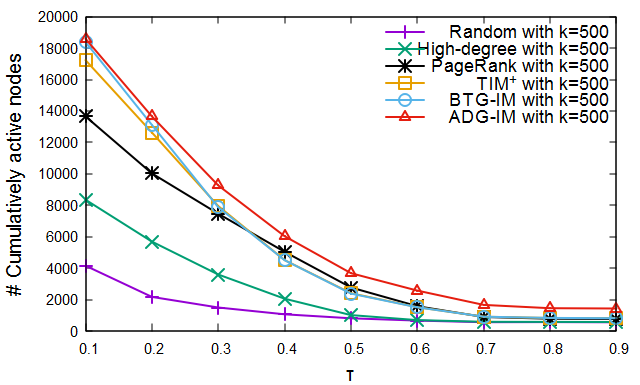}
	}
	\caption{\small Results of activation threshold $\tau$ on NetPHY}
	\label{fig:phy_tau}
\end{figure}
}
From these figures, we observe that {\sf \RHOMAX} is the best algorithm for all settings of $\tau$ and seed set size $k$.
With the increase of $\tau$, in all algorithms, the size of the cumulatively active nodes decreases rapidly.
 This is consistent with the previous description.

\textbf{Experiment results on spreading performance of SM-CA problem.}
For seed minimization problem, algorithms {\sf Random}, {\sf High-degree} and {\sf PageRank} will output a very large seed set
	to meet the target requirement $\eta$ and thus are very ineffective.
This is already demonstrated in the previous test results, and thus we only focus on the performance
	comparison  of {\sf TIM$^+$}, {\sf \FMIN} and {\sf \RHOMIN} for the SM-CA problem.
We also clarify that when $\eta$ is small, the results of SM-CA problem can be reflected by the results of IM-CA problem since we adopt the same strategies for these two problems. Thus, we only present the result on NetPHY with large enough values of $\eta$ (Figure~\ref{fig:phy_min}).
Furthermore, for large $\eta$, we notice that setting $c>1$ for the BTG strategy is no longer beneficial,
	perhaps because
	with a large number of seeds, the needs to penalize over-the-top influence (setting $c=1$) out-weighs the need of
	compensating near-the-top influence (setting $c>1$).
	It also coincides with Theorem~\ref{thm:approratio} for the case of $\eta=n$,
	so we set $c=1$ for this test.
From Figure~\ref{fig:phy_min}, we can see that {\sf \RHOMIN}outperforms {\sf TIM$^+$} and {\sf \FMIN} significantly.

\begin{figure*}[!ht]
	\centering
	\subfigure[NetPHY $\tau=0.1$]
	{
		\label{fig:phy_0.1_min}
		\includegraphics[width=0.4\textwidth]{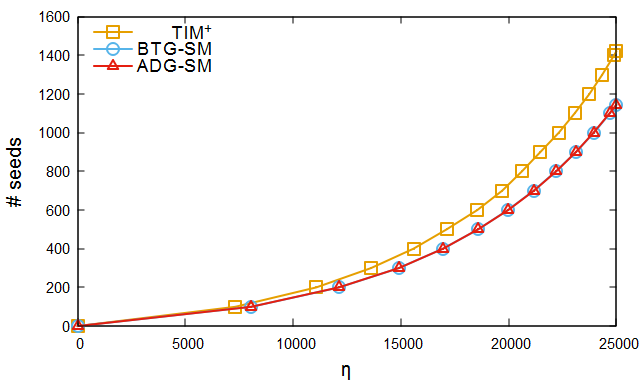}
	}
	\subfigure[NetPHY $\tau=0.3$]
	{
		\label{fig:phy_0.3_min}
		\includegraphics[width=0.4\textwidth]{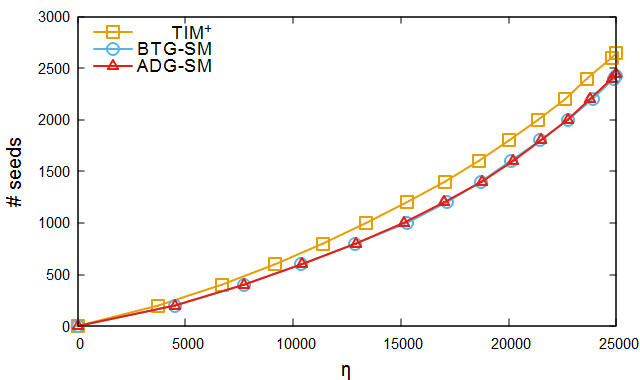}
	}

	\subfigure[NetPHY $\tau=0.5$]
	{
		\label{fig:phy_0.5_min}
		\includegraphics[width=0.4\textwidth]{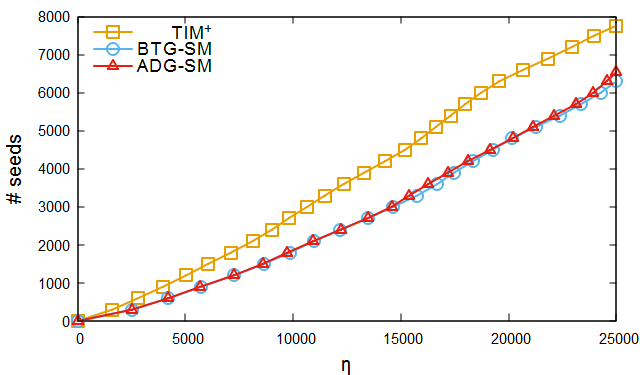}
	}
	\subfigure[NetPHY $\tau=0.7$]
	{
		\label{fig:phy_0.7_min}
		\includegraphics[width=0.4\textwidth]{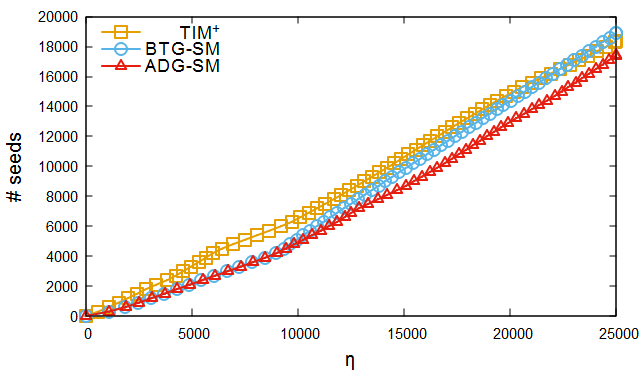}
	}
	\caption{Results of SM-CA on NetPHY}
	\label{fig:phy_min}
\end{figure*}


\textbf{Running Time.} We compare the running times of {\sf \RHOMAX}, {\sf \FMAX} and {\sf TIM$^+$} in Table \ref{table:running time}.
 We set $\tau=0.3$ and $k=500$ on all datasets.
 The results indicate that TIM$^+$ runs faster than our ADG and BTG strategies, although for the two relatively small datasets
	 the gap is not that much.
 We believe that this is because TIM$^+$ employs an estimation for the optimal influence spread with $k$ seeds for the
 influence maximization task, but we are not able to do so for the IM-CA task
 because we need to estimate individual nodes active probability
 $P_u(S)$. Nevertheless, our algorithm can still be scaled to the large graph of DBLP with millions of edges.
 Further improving the efficiency while preserving the same level of quality is a future work item.

\begin{table}[!ht]
\caption{Running Time ($\tau=0.3$, $k=500$)}
    \begin{tabular}{|c|c|c|c|}
    \hline
 &  {\sf TIM$^+$} & {\sf \RHOMAX} & {\sf \FMAX} \\ \hline
      Flixster& 39s & 87s & 138s \\ \hline
      NetPHY& 54s & 112s  & 142s \\ \hline
      DBLP & 509s & 8865s & 8685s \\ \hline
    \end{tabular}
    \label{table:running time}
\end{table}

\textbf{Conclusion.} From these experiment results,
	we conclude that ADG consistently provides the best performance cross all test cases, and thus we propose
	{\sf \RHOMAX}and {\sf \RHOMIN}as our solution to the IM-CA and SM-CA problems, respectively.
The BTG strategy performs well in some cases (such as small $\tau$ for IM-CA), but its performance is not stable, and it also
	requires tuning $c$ for different cases, and thus is less desirable than the ADG strategy.
{\sf TIM$^+$} has similar performance as the BTG strategy, but is even more unstable across the tests, and thus is not competitive
	comparing to ADG for both the SM-CA and IM-CA problems.

\section{Conclusion and Future work}
In this paper, we propose the cumulative activation influence model to reflect realistic scenarios where user adoption is based
	on repeated exposure to multiple information cascades in the network, which is different from the most existing study where
	user adoption is only based on a single cascade.
We study both the seed minimization and influence maximization problems in the cumulative activation setting, providing both
	theoretical hardness results and approximation algorithms, and further propose efficient heuristic solutions despite the
	theoretical hardness result.
Our experimental results demonstrate the effectiveness of our proposed solutions.

Our current study focuses on dealing with multiple information cascades for product adoption, but there are certainly many chances
	to further elaborate the model, such as mixing the information cascades and adoption cascades and study the impact to the
	optimization problems.
Another direction is to circumvent the theoretical hardness for general graphs due to nonsubmodularity
	and develop efficient algorithms with guarantees on
	the class of graphs that are closer to real-world graphs.

%
\bibliographystyle{abbrv}
\bibliography{inf-ca}

\end{document}